%% file: UnconstrainedRegularizedMaximization.tex
\documentclass[11pt]{article}
\usepackage{etoolbox}
\usepackage{graphicx}
\usepackage{latexsym}
\usepackage{amssymb}
\usepackage{amsmath}
\usepackage{amsthm}
\usepackage{forloop}
\usepackage[paper=letterpaper,margin=1in]{geometry}
\usepackage[ruled,vlined,linesnumbered]{algorithm2e}
\usepackage{paralist}
\usepackage{thmtools}
\usepackage{mleftright}\mleftright
\usepackage{xcolor}
\definecolor{darkgreen}{rgb}{0,0.5,0}
\usepackage{pgfplots}
\pgfplotsset{compat=1.17}
\usepgfplotslibrary{fillbetween}
\usepackage{hyperref}
\hypersetup{
    unicode=false,            % non-Latin characters in AcrobatΓÇÖs bookmarks
    colorlinks=true,          % false: boxed links; true: colored links
    linkcolor=blue!70!black,  % color of internal links (change box color with linkbordercolor)
    citecolor=darkgreen,      % color of links to bibliography
    filecolor=magenta,        % color of file links
    urlcolor=blue!70!black    % color of external links
}

\newtheorem{theorem}{Theorem}[section]
\newtheorem{lemma}[theorem]{Lemma}

\newtheorem{observation}[theorem]{Observation}

\newtheorem{reduction}{Reduction}

\newcommand{\defcal}[1]{\expandafter\newcommand\csname c#1\endcsname{{\mathcal{#1}}}}
\newcommand{\defbb}[1]{\expandafter\newcommand\csname b#1\endcsname{{\mathbb{#1}}}}
\newcommand{\defvec}[1]{\expandafter\newcommand\csname v#1\endcsname{{\mathbf{#1}}}}
\newcounter{calBbCounter}
\forLoop{1}{26}{calBbCounter}{
    \edef\Letter{\Alph{calBbCounter}}
		\edef\letter{\alph{calBbCounter}}
    \expandafter\defcal\Letter
		\expandafter\defbb\Letter
		\expandafter\defvec\letter
}

\newcommand{\eps}{\varepsilon}
\newcommand{\nnR}{{\bR_{\geq 0}}}
\newcommand{\email}[1]{{\href{mailto:#1}{#1}}}
\newcommand{\headerRef}[1]{{\texorpdfstring{\ref{#1}}{\ref*{#1}}}}
\newcommand{\characteristic}{{\mathbf{1}}}
\newcommand{\USM}{{\texttt{USM}}}
\newcommand{\RUSM}{{\texttt{RegularizedUSM}}}
\newcommand{\DGDet}{{\texttt{DeterministicDG}}}
\newcommand{\DGRand}{{\texttt{RandomizedDG}}}
\newcommand{\midset}[2]{{#1^{(#2)}}}
\newcommand{\RSet}{{\mathtt{R}}}

\author{Kobi Bodek\thanks{Department of Mathematics and Computer Science, Open University of Israel. E-mail: \email{kobibodek@gmail.com}}
				\and
				Moran Feldman\thanks{Computer Science Department, University of Haifa. E-mail: \email{moranfe@cs.haifa.ac.il}}
}
\title{Maximizing Sums of Non-monotone Submodular and Linear Functions: Understanding the Unconstrained Case}

\begin{document}

\maketitle

\pagenumbering{Alph}
\thispagestyle{empty}

\input{abstract}

\newpage
\pagenumbering{arabic}
\setcounter{page}{1}

\newtoggle{firstStatementDone}
\input{Introduction}
\toggletrue{firstStatementDone}
\input{Preliminaries}
\input{InapproximabilityMonotone}
\input{GeneralCase}
\input{NegativeL}
\input{PositiveL}

\appendix
\input{SymmetryGapAppendix}

\bibliographystyle{plain}
\bibliography{RegularizedUnconstrained}

\end{document}

%% file: Abstract.tex
\begin{abstract}
Motivated by practical applications, recent works have considered maximization of sums of a submodular function $g$ and a linear function $\ell$. Almost all such works, to date, studied only the special case of this problem in which $g$ is also guaranteed to be monotone. Therefore, in this paper we systematically study the simplest version of this problem in which $g$ is allowed to be non-monotone, namely the unconstrained variant, which we term \texttt{Regularized Unconstrained Submodular Maximization} ({\RUSM}).

Our main algorithmic result is the first non-trivial guarantee for general {\RUSM}. For the special case of {\RUSM} in which the linear function $\ell$ is non-positive, we prove two inapproximability results, showing that the algorithmic result implied for this case by previous works is not far from optimal. Finally, we reanalyze the known Double Greedy algorithm to obtain improved guarantees for the special case of {\RUSM} in which the linear function $\ell$ is non-negative; and we complement these guarantees by showing that it is not possible to obtain $(1/2, 1)$-approximation for this case (despite intuitive arguments suggesting that this approximation guarantee is natural).

\medskip

\end{abstract}

%% file: Introduction.tex
\section{Introduction} \label{sec:introduction}

The field of submodular optimization has been rapidly developing over the last two decades, partially due to new applications. Some of these applications have also motivated the optimization of composite objective functions that can be represented as the sum of a submodular function $g$ and a linear function $\ell$. Let us briefly discuss two such applications

The first application is \emph{optimization with a regularizer}. To avoid overfitting in machine-learning, it is customary to optimize a function of the form $g - \ell$, where $g$ is the quantity that we would like to maximize and $\ell$ is a (often linear) function that favors small solutions. This function $\ell$ is known as ``regularizer'' in the machine learning jargon, or ``soft constraint'' in the operations research jargon.

The other application we discuss is \emph{optimization with a curvature}. Traditionally, the theoretical study of submodular optimization problems looks for approximation guarantees that apply to all submodular functions, or at least all monotone submodular functions. However, approximation guarantees of this kind are often pessimistic, and do not capture the practical performance of the algorithms analyzed. This has motivated studying how the optimal approximation ratios of various submodular maximization problems depend on various numerical function properties. Historically, the first property of this kind to be defined was the \emph{curvature} property, which was suggested by Conforti and Cornu{\'{e}}jol~\cite{conforti1984submodular} already in $1984$. The curvature measures the distance of the submodular function from being linear, and a strong connection was demonstrated by  Sviridenko et al.~\cite{sviridenko2017optimal} between optimizing a submodular function with a given curvature and optimizing the sum $g + \ell$ of a monotone submodular function $g$ and a linear function $\ell$.

Motivated by the above applications, Sviridenko et al.~\cite{sviridenko2017optimal} also initialized the study of the optimization of $g + \ell$ sums. In particular, they described algorithms with optimal approximation guarantees for this problem when $g$ is a non-negative monotone submodular function, $\ell$ is a linear function and the optimization is subject to either a matroid or a cardinality constraint.\footnote{Technically, Sviridenko et al.~\cite{sviridenko2017optimal} proved optimal approximation guarantees only for the case in which the coefficient $\beta$ of $\ell$ is $1$ (see details below). However, their results were extended to the general case of $\beta \geq 0$ by Feldman~\cite{feldman2021guess}.} Later works obtained faster and semi-streaming algorithms for the same setting~\cite{feldman2021guess,harshaw2019submodular,kazemi2021regularized,nikolakaki2021efficient}. However, in contrast to all these (often tight) results for monotone submodular functions $g$, much less is known about the case of non-monotone submodular functions. In fact, we are only aware of a single previous work that considered $g + \ell$ sums involving such functions~\cite{lu2021regularized}.\footnote{Very recently, another work of this kind appeared as a pre-print~\cite{sun2022maximizing}. However, the main result of~\cite{sun2022maximizing} is identical to the result of~\cite{lu2021regularized}. In particular, it is important to note that the result of~\cite{sun2022maximizing} applies only to non-positive $\ell$ functions, like the result of~\cite{lu2021regularized}, although this is not explicitly stated in~\cite{sun2022maximizing}.}

Given the rarity of results so far for optimizing $g + \ell$ with a function $g$ that is non-monotone, this paper is devoted to a systematic study of the simplest problem of this kind, namely, unconstrained maximization of such sums. Formally, we study the \texttt{Regularized Unconstrained Submodular Maximization} (\RUSM) problem. In this problem, we are given a non-negative submodular function $g \colon 2^\cN \to \nnR$ and a linear function $\ell \colon 2^\cN \to \bR$ over the same ground set $\cN$, and the objective is to output a set $T \subseteq \cN$ maximizing the sum $g(T) + \ell(T)$. Unfortunately, it is not possible to prove standard multiplicative approximation ratios for {\RUSM} (implied, e.g., by Theorem~\ref{thm:negative_inapproximability}). Therefore, we follow previous works, and look in this work for algorithms that output a (possibly randomized) set $T \subseteq \cN$ such that $\bE[g(T) + \ell(T)] \geq \max_{S \subseteq \cN} [\alpha \cdot g(S) + \beta \cdot \ell(S)]$ for some coefficients $\alpha, \beta \geq 0$. For convenience, we say that an algorithm having this guarantee is an $(\alpha, \beta)$-approximation algorithm.\footnote{Some previous works compare their algorithms against $\alpha \cdot g(OPT) + \beta \cdot \ell(OPT)$, where $OPT$ is a feasible set maximizing $g(OPT) + \ell(OPT)$, instead of comparing against $\max_{S \subseteq \cN} [\alpha \cdot g(S) + \beta \cdot \ell(S)]$ like we do in this paper. This distinction is usually of little consequence.}

It is instructive to begin the study of {\RUSM} with the special case in which the objective function $g$ is guaranteed to be monotone (in addition to being non-negative and submodular). We refer below to this special case as ``monotone {\RUSM}''. The work of Feldman~\cite{feldman2021guess} on constrained maximization of $g + \ell$ immediately implies $(1 - e^{-\beta}, \beta)$-approximation for monotone {\RUSM} for every $\beta \in [0, 1]$. Our first result provides a matching inapproximability result.

\begin{restatable}{theorem}{thmMonotoneInapproximability} \label{thm:monotone_inapproximability}
For every $\beta \geq 0$ and $\eps > 0$, no polynomial time algorithm can guarantee $(1 - e^{-\beta} + \eps, \beta)$-approximation for monotone {\RUSM} even when the linear function $\ell$ is guaranteed to be non-positive.
\end{restatable}

We would like to draw attention to two properties of Theorem~\ref{thm:monotone_inapproximability}. First, for $\beta = 1$ the coefficient of $g$ in the inapproximability proved by the theorem is $1 - 1/e$, matching the optimal approximation ratio for the problem of maximizing a monotone submodular function subject to a matroid constraint. Therefore, in a sense, adding the linear part $\ell$ makes the unconstrained problem as hard as this constrained problem. Interestingly, we get a similar result for {\RUSM} below.

The other noteworthy property of Theorem~\ref{thm:monotone_inapproximability} is that it applies to any $\beta \geq 0$, while the algorithmic result of Feldman~\cite{feldman2021guess} applies only to $\beta \in [0, 1]$. This difference between the results exists because, when $\ell$ can take positive values, setting the coefficient $\beta$ to be larger than $1$ might require the algorithm to output a set $T \subseteq \cN$ obeying $\ell(T) > \max_{S \subseteq \cN} \ell(S)$. However, it turns out that, when $\ell$ is non-positive, the algorithmic result can be extended to match Theorem~\ref{thm:monotone_inapproximability} for every $\beta \geq 0$. To understand how this can be done, we need to discuss the previous work in a bit more detail.

Sviridenko et al.~\cite{sviridenko2017optimal} designed two algorithms for maximizing $g + \ell$ sums, one of which was based on the continuous greedy algorithm of C{\u{a}}linescu et al.~\cite{calinescu2011maximizing}. It is possible modify this algorithm to be based instead on a related algorithm called ``measured continuous greedy'' due to~\cite{feldman2011unified}. In general, this does not lead to any result for maximizing $g + \ell$ sums. However, Lu et al.~\cite{lu2021regularized} recently observed that one can obtain in this way results when $\ell$ is non-positive. In particular, it leads to $(1 - e^{-\beta}, \beta)$-approximation for the special case of monotone {\RUSM} in which $\ell$ is non-positive for any constant $\beta \geq 0$, which settles the approximability of monotone {\RUSM}.

We now get to the study of (not necessarily monotone) {\RUSM}. The only result that is known to date for this problem is $(1/e, 1)$-approximation for the special case in which $\ell$ is non-positive, which was proved by Lu et al.~\cite{lu2021regularized} using the technique discussed above. Our main algorithmic contribution is the first algorithm with a non-trivial approximation guarantee for general {\RUSM}.
\begin{restatable}{theorem}{thmGeneral} \label{thm:general}
For every constant $\beta \in (0, 1]$, let us define $\alpha(\beta) = \beta(1 - \beta) / (1+\beta)$. Then, for every constant $\eps \in (0, \alpha(\beta))$, there exists a polynomial time $(\alpha(\beta) - \eps, \beta - \eps)$-approximation algorithm for {\RUSM}.
\end{restatable}

We also study in more detail the special cases of {\RUSM} in each $\ell$ is either non-negative or non-positive. The above mentioned result of Lu et al.~\cite{lu2021regularized} for {\RUSM} with a non-positive $\ell$ can be extended (using the ideas of Feldman~\cite{feldman2021guess}) to get $(\beta e^{-\beta}, \beta)$-approximation for the same special case for any $\beta \in [0, 1]$.\footnote{Technically, this result can be extended to any constant $\beta \geq 0$, but this is not interesting since $\beta e^{-\beta}$ is a decreasing function for $\beta \geq 1$.} It is not immediately clear, however, how good this extended result is. For example, one can compare it with the inapproximability result of Theorem~\ref{thm:monotone_inapproximability} (which applies to the current setting as well), but there is a large gap between the above algorithmic and inapproximability results when the $\beta$ coefficient of $\ell$ is relatively large (see Figure~\ref{fig:negative_L}). This gap exists because Theorem~\ref{thm:monotone_inapproximability} holds even in the special case in which $g$ is monotone. Therefore, we prove the following theorem, which provides an alternative inapproximability result designed for the non-monotone case. Since it is difficult to understand the behavior of the expression stated in Theorem~\ref{thm:negative_inapproximability}, we numerically draw it in Figure~\ref{fig:negative_L}, which demonstrates that Theorem~\ref{thm:negative_inapproximability} closes much of the gap left with regard to {\RUSM} with non-positive linear function $\ell$.

\begin{figure}
\begin{center}\input{NegativeLGraph.tikz}\end{center}
\caption{Graphical presentation of the existing results for {\RUSM} with a non-positive linear function $\ell$. The $x$ and $y$ axes represent the coefficients of $\ell$ and $g$, respectively. The algorithmic guarantee drawn is the $(\beta e^{-\beta}, \beta)$-approximation obtainable by generalizing Lu et al.~\cite{lu2021regularized}. The shaded area represents the gap that still exists between the best known approximation guarantee and inapproximability results.} \label{fig:negative_L}
\end{figure}
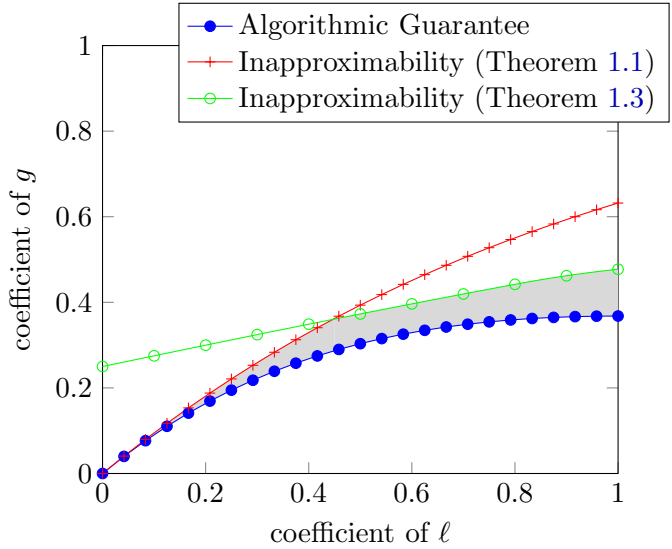

\begin{restatable}{theorem}{thmNegativeInapproximability} \label{thm:negative_inapproximability}
Given a value $\beta \geq 0$, let us define
\[
	\alpha(\beta) = \min_{\substack{t \geq 1\\r \in (0, 1/2]}} \left\{\frac{t + 1 + \sqrt{(t + 1)^2 - 8t r}}{4t} - \frac{r}{t + 1} \cdot \left[1 - \beta - 2\ln\left(\frac{t + 1- \sqrt{(t + 1)^2 - 8t r}}{2}\right)\right] \right\}
	\enspace.
\]
Then, for every $\eps > 0$, no polynomial time algorithm can guarantee $(\alpha(\beta) + \eps, \beta)$-approximation for {\RUSM} even when the linear function $\ell$ is guaranteed to be non-positive.
\end{restatable}

It is interesting to note that, for $\beta = 1$, Theorem~\ref{thm:negative_inapproximability} matches the state-of-the-art inapproximability result of Oveis Gharan and Vondr\'{a}k~\cite{ovis_gharan2011submodular} for maximizing a non-negative submodular function subject to matroid constraint. Therefore, at least at the level of the known inapproximability results, {\RUSM} with a non-positive $\ell$ is as hard as maximizing a non-negative submodular function subject to a matroid constraint.

It remains to consider the special case of {\RUSM} with a non-negative $\ell$. Here $g + \ell$ is a non-negative submodular function on its own right, and therefore, {\RUSM} becomes a special case of the well-studied problem of \text{Unconstrained Submodular Maximization} (\USM). The optimal approximation ratio for {\USM} is $1/2$ due to an inapproximability result of Feige et al.~\cite{feige2011maximizing}, and the first algorithm to obtain this approximation ratio was the ``Double Greedy'' algorithm of Buchbinder et al.~\cite{buchbinder2015tight}. Specifically, Buchbinder et al.~\cite{buchbinder2015tight} described two variants of their algorithm, a deterministic variant guaranteeing $1/3$-approximation, and a randomized variant guaranteeing $1/2$-approximation. We refer below to these two variants as {\DGDet} and {\DGRand}, respectively. Interestingly, we are able to show in the next two theorems that the performance of {\DGDet} and {\DGRand} for {\RUSM} is even better than what one would expected based on the guarantees of these algorithms for general {\USM}. %We note also that both {\DGDet} and {\DGRand} are oblivious to the value of $\alpha$. 

\begin{restatable}{theorem}{thmDeterministic} \label{thm:deterministic}
When $\ell$ is non-negative, {\DGDet} is an $(\alpha, 1 - \alpha)$-approximation algorithm for {\RUSM} for all $\alpha \in [0, 1/3]$ at the same time (the algorithm is oblivious to the value of $\alpha$).
\end{restatable}

\begin{restatable}{theorem}{thmRandomized} \label{thm:randomized}
When $\ell$ is non-negative, {\DGRand} is an $(\alpha, 1 - \alpha/2)$-approximation algorithm for {\RUSM} for all $\alpha \in [0, 1/2]$ at the same time (the algorithm is oblivious to the value of $\alpha$).
\end{restatable}

We conclude this section with an interesting observation. Up to this point, the most well studied $g + \ell$ maximization problem was maximizing the sum of a non-negative monotone submodular function $g$ and a linear function $\ell$ subject to a matroid constraint. When $\ell$ is positive, the optimal approximation guarantee for this problem is $(1 - 1/e, 1)$~\cite{sviridenko2017optimal}, which is natural since $1 - 1/e$ is the optimal approximation ratio for maximizing such a function $g$ subject to a matroid constraint~\cite{nemhauser1978best}. Thus, one might expect to get $(1/2, 1)$-approximation for {\RUSM} with a non-negative $\ell$. However, both Theorems~\ref{thm:deterministic} and~\ref{thm:randomized} fail to prove such a guarantee, and we are able to show that this is not a coincidence.
\begin{restatable}{theorem}{thmPositiveInapproximability} \label{thm:positive_inapproximability}
Even when the linear function $\ell$ is guaranteed to be non-negative, no polynomial time algorithm can guarantee $(1/2, 1)$-approximation for {\RUSM}.
\end{restatable}

\paragraph{Paper Structure.} In Section~\ref{sec:preliminaries} we give a few formal definitions and explain the notation used throughout the paper. Then, we prove our inapproximability result for monotone {\RUSM} (Theorem~\ref{thm:monotone_inapproximability}) in Section~\ref{sec:inapproximability_monotone}. Our results for general {\RUSM}, {\RUSM} with non-positive $\ell$ and {\RUSM} with non-negative $\ell$ can be found in Sections~\ref{sec:general}, \ref{sec:negative} and~\ref{sec:positive}, respectively.

%% file: NegativeLGraph.tikz
\begin{tikzpicture}
\begin{axis}[
    xlabel = {coefficient of $\ell$},
    ylabel = {coefficient of $g$},
    xmin=0, xmax=1,
    ymin=0, ymax=1,
		legend cell align=left,
		legend style={at={(1.1,1.1)}}]
 
% Plot 1
\addplot [name path = algorithm, blue, mark = *, domain = 0:1] {x * exp(-x)};
\addlegendentry{Algorithmic Guarantee}
 
% Plot 2
\addplot [name path = inapproximability_monotone, red, mark = +, domain = 0:1] {1 - exp(-x)};
\addlegendentry{Inapproximability (Theorem~\ref{thm:monotone_inapproximability})}
 
% Plot 3
\addplot [name path = inapproximability_nonmonotone, green, mark=o, mark repeat=10] coordinates {
        (0.000000, 0.250030)
        (0.010000, 0.252525)
        (0.020000, 0.255020)
        (0.030000, 0.257515)
        (0.040000, 0.260010)
        (0.050000, 0.262505)
        (0.060000, 0.265000)
        (0.070000, 0.267495)
        (0.080000, 0.269990)
        (0.090000, 0.272485)
        (0.100000, 0.274980)
        (0.110000, 0.277475)
        (0.120000, 0.279966)
        (0.130000, 0.282456)
        (0.140000, 0.284946)
        (0.150000, 0.287435)
        (0.160000, 0.289920)
        (0.170000, 0.292405)
        (0.180000, 0.294887)
        (0.190000, 0.297367)
        (0.200000, 0.299847)
        (0.210000, 0.302322)
        (0.220000, 0.304797)
        (0.230000, 0.307267)
        (0.240000, 0.309737)
        (0.250000, 0.312203)
        (0.260000, 0.314668)
        (0.270000, 0.317129)
        (0.280000, 0.319589)
        (0.290000, 0.322044)
        (0.300000, 0.324497)
        (0.310000, 0.326947)
        (0.320000, 0.329393)
        (0.330000, 0.331838)
        (0.340000, 0.334278)
        (0.350000, 0.336715)
        (0.360000, 0.339150)
        (0.370000, 0.341580)
        (0.380000, 0.344006)
        (0.390000, 0.346430)
        (0.400000, 0.348850)
        (0.410000, 0.351265)
        (0.420000, 0.353677)
        (0.430000, 0.356086)
        (0.440000, 0.358491)
        (0.450000, 0.360891)
        (0.460000, 0.363287)
        (0.470000, 0.365680)
        (0.480000, 0.368068)
        (0.490000, 0.370453)
        (0.500000, 0.372833)
        (0.510000, 0.375208)
        (0.520000, 0.377579)
        (0.530000, 0.379945)
        (0.540000, 0.382308)
        (0.550000, 0.384665)
        (0.560000, 0.387019)
        (0.570000, 0.389367)
        (0.580000, 0.391711)
        (0.590000, 0.394051)
        (0.600000, 0.396385)
        (0.610000, 0.398715)
        (0.620000, 0.401040)
        (0.630000, 0.403360)
        (0.640000, 0.405675)
        (0.650000, 0.407985)
        (0.660000, 0.410290)
        (0.670000, 0.412590)
        (0.680000, 0.414884)
        (0.690000, 0.417174)
        (0.700000, 0.419459)
        (0.710000, 0.421738)
        (0.720000, 0.424012)
        (0.730000, 0.426281)
        (0.740000, 0.428544)
        (0.750000, 0.430802)
        (0.760000, 0.433055)
        (0.770000, 0.435303)
        (0.780000, 0.437545)
        (0.790000, 0.439781)
        (0.800000, 0.442012)
        (0.810000, 0.444238)
        (0.820000, 0.446437)
        (0.830000, 0.448582)
        (0.840000, 0.450673)
        (0.850000, 0.452710)
        (0.860000, 0.454694)
        (0.870000, 0.456626)
        (0.880000, 0.458507)
        (0.890000, 0.460337)
        (0.900000, 0.462117)
        (0.910000, 0.463848)
        (0.920000, 0.465529)
        (0.930000, 0.467163)
        (0.940000, 0.468749)
        (0.950000, 0.470288)
        (0.960000, 0.471781)
        (0.970000, 0.473228)
        (0.980000, 0.474631)
        (0.990000, 0.475989)
        (1.000000, 0.477302)
};
\addlegendentry{Inapproximability (Theorem~\ref{thm:negative_inapproximability})}

% Fill area between algorithmic result and inapproximabilities
\addplot [black!15] fill between [of = algorithm and inapproximability_nonmonotone, soft clip = {domain=0.45:1}];
\addplot [black!15] fill between [of = algorithm and inapproximability_monotone, soft clip = {domain=0:0.45}];

\end{axis}
\end{tikzpicture}

%% file: Preliminaries.tex
\section{Preliminaries} \label{sec:preliminaries}

\paragraph{Set Functions and Notation.} Given a set function $f\colon 2^\cN \to \bR$, an element $u \in \cN$ and a set $S \subseteq \cN$, the marginal contribution of $u$ to $S$ with respect to $f$ is $f(u \mid S) \triangleq f(S \cup \{u\}) - f(S)$. A set function $f\colon 2^\cN \to \bR$ is called \emph{submodular} if it satisfies the intuitive property of diminishing returns. More formally, $f$ is submodular if $f(u \mid S) \geq f(u \mid T)$ for every two sets $S \subseteq T \subseteq \cN$ and element $u \in \cN \setminus T$. An equivalent definition of submodularity is that $f$ is submodular if $f(S) + f(T) \geq f(S \cup T) + f(S \cap T)$ for every two sets $S, T \subseteq \cN$.

The set function $f$ is called \emph{monotone} if $f(S) \leq f(T)$ for every two sets $S \subseteq T \subseteq \cN$, and it is called \emph{linear} if there exist values $\{a_u \in \bR \mid u \in \cN\}$ such that $f(S) = \sum_{u \in S} a_u$ for every set $S \subseteq \cN$.\footnote{Linear set functions are also known as \emph{modular} functions.} One can verify that any linear set function is submodular, but the reverse does not necessarily hold. Additionally, given a set $S$, a set function $f$ and an element $u$, we often use $S + u$, $S - u$ and $f(u)$ as shorthands for $S \cup \{u\}$, $S \setminus \{u\}$ and $f(\{u\})$, respectively.

\paragraph{Multilinear extension.} It is often useful to consider continuous extensions of set functions, and there are multiple ways in which this can be done. The proofs of our inapproximability results employ one such extension known as the \emph{multilinear extension} (due to~\cite{calinescu2011maximizing}). Formally, given a set function $f\colon 2^\cN \to \bR$, its multilinear extension is the function $F\colon [0, 1]^\cN \to \bR$ defined, for every vector $\vx \in [0, 1]^\cN$, by $F(\vx) = \bE[f(\RSet(\vx))]$, where $\RSet(\vx)$ is a random subset of $\cN$ including every element $u \in \cN$ with probability $x_u$, independently.

One can verify that, as is suggested by its name, the multilinear extension $F$ is a multilinear function of the coordinates of its input vector. Furthermore, $F$ is an extension of the set function $f$ in the sense that for every set $S \subseteq \cN$ we have $F(\characteristic_S) = f(S)$, where $\characteristic_S$ is the characteristic vector of the set $S$ (i.e., a vector that has the value $1$ in coordinates corresponding to elements of $S$, and the value $0$ in the other coordinates).

\paragraph{Value Oracle.} As is standard in the submodular optimization literature, we assume in this paper that algorithms access their set function inputs only through \emph{value oracles}. A value oracle for a set function $f$ is a black box that given a set $S \subseteq \cN$ returns $f(S)$. One advantage of this convention is that it makes it possible to use information theoretic arguments to prove unconditional inapproximability results (i.e., inapproximability results that are not based on any complexity assumption). Nevertheless, if necessary, these inapproximability results can usually be adapted to apply also to succinctly represented functions (instead of functions accessed via value oracles) at the cost of introducing some complexity assumption~\cite{dobzinski2012query}.

%% file: InapproximabilityMonotone.tex
\section{Inapproximability for Monotone Functions} \label{sec:inapproximability_monotone}

In this section we show an inapproximability for monotone {\RUSM} (Theorem~\ref{thm:monotone_inapproximability}). All our inapproximability results in this paper are proved using Theorem~\ref{thm:symmetry_gap}. Since the proof of this theorem is a relatively straightforward adaptation of the symmetry gap framework of Vondr\'{a}k~\cite{vondrak2013symmetry}, we defer it to Appendix~\ref{app:symmetry_gap}. 

\begin{restatable}{theorem}{thmSymmetryGap} \label{thm:symmetry_gap}
Consider an instance $(g, \ell)$ of {\RUSM} consisting of a non-negative submodular function $g \colon 2^\cN \to \nnR$ and a linear function $\ell \colon 2^\cN \to \nnR$, and assume that there exists a group $\cG$ of permutations over $\cN$ such that the equalities $g(S) = g(\sigma(S))$ and $\ell(S) = \ell(\sigma(S))$ hold for all sets $S \subseteq \cN$ and permutations $\sigma \in \cG$. Let $G$ and $L$ be the multilinear extensions of $g$ and $\ell$ respectively, and for every vector $\vx \in [0, 1]^\cN$, let us denote $\bar{\vx} = \bE_{\sigma \in \cG}[\vx]$, i.e., $\bar{\vx}$ is the expected vector $\sigma(\vx)$ when $\sigma$ is picked uniformly at random out of $\cG$. For any two constants $\alpha, \beta \geq 0$, if $\max_{S \subseteq \cN} [\alpha \cdot g(S) + \beta \cdot \ell(S)]$ is strictly positive and
\[
	\max_{\vx \in [0, 1]^\cN} [G(\bar{\vx}) + L(\bar{\vx})] \leq \max_{S \subseteq \cN} [\alpha \cdot g(S) + \beta \cdot \ell(S)]
	\enspace,
\]
then no polynomial time algorithm for {\RUSM} can guarantee $((1 + \eps)\alpha, (1 + \eps)\beta)$-appro\-ximation for any positive constant $\eps$. Furthermore, this inapproximability guarantee holds also when we restrict attention to instances $(g', \ell')$ of {\RUSM} having the following additional properties.
\begin{itemize}
	\item If $\ell$ is non-negative or non-positive, then we can assume that $\ell'$ also has the same property.
	\item If $g$ is monotone, then we can assume that $g'$ is monotone as well.
\end{itemize}
\end{restatable}

In the common case in which the linear function $\ell$ is a non-positive, the following observation allows us to produce slightly cleaner results using Theorem~\ref{thm:symmetry_gap}.
\begin{observation} \label{obs:negative_simplification}
If $\ell$ is non-positive and $\alpha > 0$, then one can replace the term ``$((1 + \eps)\alpha, (1 + \eps)\beta)$-approximation'' in Theorem~\ref{thm:symmetry_gap} with the term ``$(\alpha + \eps, \beta)$-approximation''.
\end{observation}
\begin{proof}
Recall that Theorem~\ref{thm:symmetry_gap} proves, under some conditions, that no polynomial time algorithm for {\RUSM} has $((1 + \eps)\alpha, (1 + \eps)\beta)$-approximation. Furthermore, if we reduce the value of the constant parameter $\eps$ of the theorem by a factor of $\alpha$, then the theorem also shows that no such algorithm can guarantee $(\alpha + \eps, \beta + \eps\beta/\alpha)$-approximation. This implies the observation since, when $\ell$ is non-positive, any $(\alpha + \eps, \beta)$-approximation algorithm for {\RUSM} is also an $(\alpha + \eps, \beta + \eps\beta/\alpha)$-approximation algorithm.
\end{proof}

To prove Theorem~\ref{thm:monotone_inapproximability} using Theorem~\ref{thm:symmetry_gap}, we need to define an instance $\cI$ of monotone {\RUSM}. Specifically, consider a ground set $\cN$ of size $n \geq 2$ and a value $r \in (0, 1]$, and let us define
\[
	g(S) = \min\{|S|, 1\}
	\quad
	\text{and}
	\quad
	\ell(S) = -r \cdot |S|
	\qquad
	\forall\;S \subseteq \cN
	\enspace.
\]

\begin{lemma} \label{lem:monotone_symmetry_gap}
For any constants $\eps > 0$, $\beta \geq 0$ and $\alpha = 1 - e^{-\beta} + \eps$, when $n$ is large enough, there exists a value $r \in (0, 1]$ such that the inequality of Theorem~\ref{thm:symmetry_gap} applies to $\cI$ and $\max_{S \subseteq \cN} [\alpha \cdot g(S) + \beta \cdot \ell(S)]$ is strictly positive.
\end{lemma}
\begin{proof}
Observe that $\max_{S \subseteq \cN} [\alpha \cdot g(S) + \beta \cdot \ell(S)] \geq \alpha - \beta r = 1 - e^{-\beta} + \eps - \beta r$ because $S$ can be chosen as a singleton subset of $\cN$. Let us now study the left hand side of the inequality of Theorem~\ref{thm:symmetry_gap}. Since both $g$ and $\ell$ are unaffected when an arbitrary permutation is applied to the ground set, we can choose $\cG$ as the group of all permutations over $\cN$. Thus, for every vector $\vx \in [0, 1]^\cN$,
\[
	\bar{\vx} = \frac{\|\vx\|_1}{n} \cdot \characteristic_\cN
	\enspace.
\]
Therefore,
\begin{align*}
	\max_{\vx \in [0, 1]^\cN} [G(\bar{\vx}) + L(\bar{\vx})]
	={} &
	\max_{x \in [0, 1]} [G(x \cdot \characteristic_\cN) + L(x \cdot \characteristic_\cN)]\\
	={} &
	\max_{x \in [0, 1]} [1 - (1 - x)^n - xrn]
	=
	1 - r - r(n - 1)[1 - r^{1 / (n - 1)}]
	\enspace,
\end{align*}
where the last equality holds since the maximum is obtained for $x = 1 - \sqrt[n - 1]{r}$. Note now that if we denote $y = (n - 1)^{-1}$, then by L'H\^{o}pital's rule,
\[
	\lim_{n \to \infty} (n - 1)[1 - r^{1 / (n - 1)}]
	=
	\lim_{y \to 0} \frac{1 - r^y}{y}
	=
	\lim_{y \to 0} \frac{- r^y\ln r}{1}
	=
	- \ln r
	\enspace,
\]
and therefore, for a large enough $n$, $(n - 1)[1 - r^{1 / (n - 1)}] \geq - \ln r - \eps$; which implies
\[
	\max_{\vx \in [0, 1]^\cN} [G(\bar{\vx}) + L(\bar{\vx})]
	\leq
	1 - r - r[- \ln r - \eps]
	\leq
	1 - r(1 - \ln r) + \eps.
\]

Given the above bounds, we get that the inequality of Theorem~\ref{thm:symmetry_gap} holds for any $r > 0$ obeying
\[
	1 - e^{-\beta} + \eps - \beta r
	\geq
	1 - r(1 - \ln r) + \eps
	\enspace.
\]
Since the last inequality is equivalent to
\[
	r - r \ln r \geq e^{-\beta} + \beta r
	\enspace,
\]
it holds for $r = e^{-\beta} \subseteq (0, 1]$. Furthermore, for this choice of $r$,
\[
	\max_{S \subseteq \cN} [\alpha \cdot g(S) + \beta \cdot \ell(S)]
	\geq
	1 - e^{-\beta} + \eps - \beta r
	=
	1 - (1 + \beta)e^{-\beta} + \eps
	\geq
	1 - \frac{1 + \beta}{1 + \beta} + \eps
	=
	\eps
	>
	0
	\enspace.
	\qedhere
\]
\end{proof}

Theorem~\ref{thm:monotone_inapproximability}, which we repeat here for convenience, now follows by combining Theorem~\ref{thm:symmetry_gap}, Observation~\ref{obs:negative_simplification} and Lemma~\ref{lem:monotone_symmetry_gap} since $g$ is a non-negative monotone submodular function and $\ell$ is a non-positive linear function.
\thmMonotoneInapproximability*

%% file: GeneralCase.tex
\section{Algorithm for the General Case} \label{sec:general}

In this section we describe and analyze the only non-trivial algorithm known to date (as far as we know) for general {\RUSM}. Using this algorithm we prove Theorem~\ref{thm:general}, which we repeat here for convenience.
\thmGeneral*

Our algorithm is based on a non-oblivious local search, i.e., a local search guided by an auxiliary function rather than the objective function. Non-oblivious local searches have been used previously in the context of submodular maximization by, for example, Feige et al.~\cite{feige2011maximizing} and Filmus and Ward~\cite{filmus2014monotone}. The auxiliary function used by our algorithm is a function $h \colon 2^\cN \to \nnR$ defined as follows. For every set $S \subseteq \cN$,
\[
	h(S) = \bE[g(S(\beta))] + \beta(1 + \beta) \cdot \ell(S)
	\enspace,
\]
where $S(\beta)$ is a random subset of $S$ that includes every element of $S$ with probability $\beta$, independently.

Ideally, we would like to find a local maximum with respect to $h$, i.e., a set $T \subseteq \cN$ such that the value of $h(T)$ cannot be increased either by adding a single element to $T$, or by removing a single element from $T$. However, there are two issues that make the task of finding such a local maximum difficult.
\begin{itemize}
	\item We do not know how to exactly evaluate the expectation in the definition of $h$ in polynomial time. Therefore, whenever we need to calculate expressions involving $h$, we have to approximate them using sampling, which introduces estimation errors that have to be taken into account.
	\item A straightforward local search algorithm changes its current solution whenever adding or removing a single element improves this solution. However, the time complexity of such a na\"{i}ve algorithm can be exponential. Therefore, our algorithm adds or removes an element only when this is beneficial enough, which means that the algorithm finds an approximate local maximum rather than a true one. Employing this idea is not trivial given the errors introduced by the sampling, as mentioned above. However, we manage to prove that, for the value $\Delta$ defined by our algorithm, with high probability: (i) the algorithm only makes changes that increase the value of $h(T)$ by $\Delta/2$ or more, and (ii) the algorithm continuous to make changes as long as there exists some possible change that increases the value of $h(T)$ by at least $3\Delta/2$.
	
The quality of the approximate local maximum produced by our algorithm is controlled by the parameter $\eps$ of Theorem~\ref{thm:general}. Setting a lower value for $\eps$ decreases $\Delta$, which increases the time complexity of our algorithm, but also makes the approximate local maximum produced closer to being a true local maximum, and thus, improves the approximation guarantee.
\end{itemize}

Let $\hat{S}$ be a subset of $\cN$ maximizing $\left(\alpha(\beta) - \eps\right) \cdot g(\hat{S}) + (\beta - \eps) \cdot \ell(\hat{S})$. To implement the solutions described in the last two bullets, it is useful to assume that the ground set $\cN$ does not include elements that have some problematic properties. The following reduction shows that we can assume that this is indeed the case without loss of generality.

\begin{reduction} \label{red:remove_elements}
While proving Theorem~\ref{thm:general}, we may assume that every element $u \in \cN$ obeys
\[
	\alpha(\beta) \cdot g(u) + \beta \cdot \ell(u) \geq 0
	\qquad
	\text{and}
	\qquad
	\max\{g(u) + \ell(u), g(\varnothing)\} \leq \beta \cdot [g(\hat{S}) + \ell(\hat{S})]
	\enspace.
\]
\end{reduction}
\begin{proof}
Let us begin by proving the second part of the reduction. If there exists an element $u$ for which
\[
	\max\{g(u) + \ell(u), g(\varnothing) + \ell(\varnothing)\}
	=
	\max\{g(u) + \ell(u), g(\varnothing)\}
	>
	\beta \cdot [g(\hat{S}) + \ell(\hat{S})]
	\enspace,
\]
then one can obtain an algorithm with the guarantee stated in Theorem~\ref{thm:general} by simply returning the set $T \in \{\varnothing\} \cup \{\{u\} \mid u \in \cN\}$ maximizing $g(T) + \ell(T)$ because
\begin{align*}
	g(T) &{}+ \ell(T)
	\geq
	\max\{g(u) + \ell(u), g(\varnothing) + \ell(\varnothing)\}
	>
	\beta \cdot [g(\hat{S}) + \ell(\hat{S})]\\
	={} &
	\frac{\beta}{\beta - \eps} \cdot [ (\beta - \eps) \cdot g(\hat{S}) + (\beta - \eps) \cdot \ell(\hat{S})]
	\geq
	\frac{\beta}{\beta - \eps} \cdot [(\alpha(\beta) - \eps) \cdot g(\hat{S}) + (\beta - \eps) \cdot \ell(\hat{S})]\\
	={} &
	\frac{\beta}{\beta - \eps} \cdot \max_{S \subseteq \cN} [(\alpha(\beta) - \eps) \cdot g(S) + (\beta - \eps) \cdot \ell(S)]
	\geq
	\max_{S \subseteq \cN} [(\alpha(\beta) - \eps) \cdot g(S) + (\beta - \eps) \cdot \ell(S)]
	\enspace.
\end{align*}
The penultimate inequality holds since $g$ is non-negative, the last equality follows from the definition of $\hat{S}$, and the last inequality holds since the maximum on both sides of this inequality is non-negative (this can be seen by choosing $S = \varnothing$). It is also worth mentioning the above algorithm, namely, outputting the set $T \in \{\varnothing\} \cup \{\{u\} \mid u \in \cN\}$ maximizing $g(T) + \ell(T)$ can be implemented to run in linear time since it only has to consider $|\cN| + 1$ candidate sets.

It remains to prove the first part of the reduction. Assume that there exists an algorithm $ALG$ that has the guarantee stated in Theorem~\ref{thm:general} for instances obeying the first part of the reduction, and let us explain how to get an algorithm that has the same guarantee for general instances of {\RUSM}. Towards this goal, let us define
\[
	\cN'
	=
	\{u \in \cN \mid \alpha(\beta) \cdot g(u) + \beta \cdot \ell(u) \geq 0\}
	\enspace.
\]
In other words, $\cN'$ is the subset of $\cN$ that includes all the elements obeying the first part of the reduction. We claim that $\hat{S} \subseteq \cN'$. If this is not true, then, by the submodularity and non-negativity of $g$, for any element $u \in \hat{S} \setminus \cN'$ we have
\begin{align*}
	(\alpha(\beta) - \eps) \cdot [g(\hat{S}) - g(\hat{S} &{}- u)] + (\beta - \eps) \cdot [\ell(\hat{S}) - \ell(\hat{S} - u)]
	\leq
	(\alpha(\beta) - \eps) \cdot g(u \mid \varnothing) + (\beta - \eps) \cdot \ell(u)\\
	\leq{} &
	(\alpha(\beta) - \eps) \cdot g(u) + (\beta - \eps) \cdot \ell(u)
	\leq
	\frac{\beta - \eps}{\beta} \cdot [\alpha(\beta) \cdot g(u) + \beta \cdot \ell(u)]
	<
	0
	\enspace,
\end{align*}
which contradicts the definition of $\hat{S}$ (the penultimate inequality holds since $\alpha(\beta) \leq \beta$ and $g$ is non-negative).

Observe now that by the definition of $\cN'$, we can execute $ALG$ on $\cN'$, which produces a set $T$ obeying
\[
	\bE[g(T) + \ell(T)]
	\geq
	\max_{S \subseteq \cN'} [(\alpha(\beta) - \eps) \cdot g(S) + (\beta - \eps) \cdot \ell(S)]
	=
	\max_{S \subseteq \cN} [(\alpha(\beta) - \eps) \cdot g(S) + (\beta - \eps) \cdot \ell(S)]
	\enspace,
\]
where the equality holds since the maximum in the rightmost side is obtained for $S = \hat{S}$, and $\hat{S}$ is a subset of $\cN'$ as we have proved above. Since one can construct $\cN'$ in linear time, executing $ALG$ on $\cN'$ is the promised algorithm that achieves the guarantee of Theorem~\ref{thm:general} without assuming the first part of the reduction.
\end{proof}

From this point until the end of the section, we denote by $n$ the size of the ground set $\cN$. We are now ready to describe our algorithm (given as Algorithm~\ref{alg:LocalSearch}). This algorithm implicitly assumes that Reduction~\ref{red:remove_elements} was applied, that $n$ is large enough and that $\max\{g(\varnothing), \max_{u \in \cN} g(u)\} > 0$.\footnote{Let us explain why the problem becomes easy if either of the last two assumptions is violated. If $n$ is bounded by a constant, it is possible to use exhaustive search to find the set $T \subseteq \cN$ maximizing $g(T) + h(T)$, and one can verify that such a set has the properties guaranteed by Theorem~\ref{thm:general}. Additionally, if $\max\{g(\varnothing), \max_{u \in \cN} g(u)\} = 0$, then the submodularity of $g$ guarantees that $g$ is the zero function, which means that we can get the guarantee of Theorem~\ref{thm:general} by outputting the set $\{u \in \cN \mid \ell(u) > 0\}$.} The algorithm maintains a solution $T$, which it updates in iterations. In each iteration, the algorithm calculates for every element $u$ an estimate $\omega_u$ of the contribution of $u$ to the $g$ component of the auxiliary function $h$. Then, Line~\ref{line:adding} of the algorithm looks for an element $u \in \cN \setminus T$ which, based on the estimate $\omega_u$, will increase $h(T)$ by $\Delta$ if added to $T$. If such an element $u$ is found, the algorithm adds it to $T$ and continues to the next iteration. Otherwise, Line~\ref{line:removing} looks for an element $u \in T$ which will increase $h(T)$ by $\Delta$ if removed from $T$ (again, based on the estimate $\omega_T$). If such an element $u$ is found, then the algorithm removes it from $T$ and continues to the next iteration. However, if both Lines~\ref{line:adding} and~\ref{line:removing} fail to find an appropriate element, the algorithm assumes that it has encountered an approximate local maximum, and terminates. Somewhat surprisingly, when this happens the algorithm outputs a sample $\hat{T}$ of $T(\beta)$ rather than the solution $T$ itself (unless the value of this sample is negative, in which case the algorithm falls back to the solution $\varnothing$). We show below that if $T$ is an approximate local maximum of the auxiliary function $h$, then $T(\beta)$ is in expectation a good solution with respect to the objective function.

\begin{algorithm}
\caption{\texttt{Non-oblivious Local Search} $(\beta, \eps)$} \label{alg:LocalSearch}
\DontPrintSemicolon
    Let $\Delta \gets \frac{\eps}{2n} \cdot \max\{g(\varnothing), \max_{u \in \cN} g(u)\}$.\\
    Let $T \gets \{u \in \cN \mid \ell(u) > 0\}$.\\
    \For{$i = 1$ \KwTo $\lceil 4n^2 / \eps \rceil  + 1$}
    {
        \lFor{every $u \in \cN$}{Let $\omega_{u}$ be an estimate of $\beta \cdot \bE[g(u \mid T(\beta) - u)]$ obtained by taking the average of $\beta \cdot g(u \mid T(\beta) - u)$ for $k = \lceil 128n^4 \eps^{-2}\beta^2 \cdot \ln(10n^4 / \eps) \rceil$ independent samples of $T(\beta)$.}
	    \lIf{there exists $u \in \cN \setminus T$ such that $\omega_u + \beta(1 + \beta) \cdot \ell(u) \geq \Delta$}{Update $T \gets T + u$.\label{line:adding}}
	    \lElseIf{there exists $u \in T$ such that $\omega_u + \beta(1 + \beta) \cdot \ell(u) \leq -\Delta$}{Update $T \gets T - u$.\label{line:removing}}
        \lElse{Exit the ``for'' loop.\label{line:exit}}
    }
    Let $\hat{T}$ be a sample of $T(\beta)$.\\
		\lIf{$g(\hat{T}) + \ell(\hat{T}) \geq 0$}{\Return $\hat{T}$.}
		\lElse{\Return $\varnothing$.}
\end{algorithm}    

It is clear that Algorithm~\ref{alg:LocalSearch} runs in polynomial time, and therefore, we concentrate in the rest of this section on proving its approximation guarantee. Algorithm~\ref{alg:LocalSearch} makes multiple estimation during its execution. We say that an estimate $\omega_{u}$ is \emph{good} if $|\omega_{u} - \bE[g(u \mid T(\beta) - u)]| \leq \Delta/2$ (for the set $T$ at the time in which the estimate was made), otherwise the estimate is \emph{bad}.
\begin{lemma} \label{lem:Estimation}
With high probability (a probability approaching $1$ when $n$ tends to infinity), all the estimates made by Algorithm~\ref{alg:LocalSearch} are good.
\end{lemma}
\begin{proof}
Consider a particular estimate $\omega_{u}$ made by Algorithm~\ref{alg:LocalSearch}, and recall that the algorithm makes this estimate by averaging $\beta \cdot g(u \mid T(\beta) - u)$ for $k$ independent samples of $T(\beta)$. Let $T_1, T_2, \dotsc, T_k$ denote the samples of $T(\beta)$ used by the algorithm, and let us define, for every integer $1 \leq i \leq k$,
    \[
        X_{i} = \frac{g(u \mid T_i - u) - \bE[g(u \mid T(\beta) - u)]}{4n^2\Delta / \eps}
				\enspace.
    \]
Note that we deterministically have $|X_{i}| \leq 1$ because for every set $S \subseteq \cN$ and element $u \in \cN$ it holds that
\[
	g(u \mid S - u) \leq g(u \mid \varnothing) = g(u) - g(\varnothing) \leq g(u) \leq \frac{2n \Delta}{\eps}
\]
(the first inequality follows from the submodularity of $g$, the second from $g$'s non-negativity, and the last from the definition of $\Delta$), and
    \begin{align*}
				g(u \mid S - u)
				={} &
				g(S+u) - g(S - u)
				\geq
				- g(S - u)\\
				\geq{} &
				-\max\{g(\varnothing), |S - u| \cdot \max_{v \in \cN} g(v)\}
				\geq
				n \cdot \max\{g(\varnothing), \max_{v \in \cN} g(v)\}
				=
        -\frac{2n^2 \Delta}{\eps}
				\enspace,
    \end{align*}
    where the first and last inequalities hold since $g$ is non-negative and the last equality follows from the definition of $\Delta$. To justify the second inequality, note that, by the submodularity and non-negativity of $g$,
		\[
			g(S - u)
			\leq
			g(\varnothing) + \sum_{v \in S - u} g(v \mid \varnothing)
			=
			(1 - |S - u|) \cdot g(\varnothing) + \sum_{v \in S - u} g(v)
			\leq
			\max\{g(\varnothing), |S - u| \cdot \max_{v \in \cN} g(v)\}
			\enspace.
		\]
		
We can now upper bound the probability that $\omega_u$ is a bad estimate as follows.
\begin{align*}
	\Pr[|\omega_u - \beta \cdot \bE[g(u \mid T(\beta) - u)]| >{}& \Delta/2]
	=
	\Pr\left[\frac{4\beta n^2\Delta}{\eps} \cdot \frac{\sum_{i = 1}^k X_i}{k} > \frac{\Delta}{2}\right]
	=
	\Pr\left[\sum_{i = 1}^k X_i > \frac{\eps k}{8\beta n^2}\right]\\
	\leq{} &
	2e^{-(\eps k / (8\beta n^2))^2/(2k)}
	=
	2e^{-\eps^2k / (128\beta^2 n^4)}
	\leq
	2^{-\ln(10n^4 / \eps)}
	=
	\frac{\eps}{5n^4}
	\enspace,
\end{align*}
where the first inequality follows from the Chernoff-like Theorem A.1.16 of~\cite{alon2000probabilistic}, and the second inequality holds by the definition of $k$.
To conclude the proof of the lemma, it remains to observe that Algorithm~\ref{alg:LocalSearch} makes at most $5n^3 / \eps$ estimates since it makes only $n$ estimates per iteration, and has at most $\lceil 4n^2 / \eps \rceil + 1 \leq 5n^2 / \eps$ iterations (the inequality holds for a large enough $n$). Therefore, by the union bound, the probability that any estimate made by this algorithm is bad can be upper bounded by $1/n$.
\end{proof}

Using the previous lemma, we can now prove that, with high probability, Algorithm~\ref{alg:LocalSearch} terminate with $T$ being an approximate local maximum.

\begin{lemma} \label{lem:Properties}
With high probability, when Algorithm~\ref{alg:LocalSearch} terminates we have 
\[
	h(T) \geq h(T + u) - 3\Delta/2 \quad \forall\; u \in \cN \setminus T
	\qquad
	\text{and}
	\qquad
	h(T) \geq h(T - u) - 3\Delta/2 \quad \forall\; u \in T
	\enspace.
\]
\end{lemma}
\begin{proof}
We prove that the lemma holds deterministically when all the estimates made by Algorithm~\ref{alg:LocalSearch} are good, which is a high probability event by Lemma~\ref{lem:Estimation}. Our first step is to show that given this assumption the value of $h(T)$ increases by at least $\Delta/2$ following every iteration of the algorithm unless this iteration terminates by Line~\ref{line:exit} (and therefore, does not modify $T$). If Algorithm~\ref{alg:LocalSearch} added an element $u \in \cN \setminus T$ to $T$ during the iteration, then the new value of $h(T)$ is
\begin{align*}
	h(T + u)
	={} &
	h(T) + \beta \cdot \bE[g(u \mid T(\beta) - u)] + \beta(1 + \beta) \cdot \ell(u)\\
	\geq{} &
	h(T) + (\omega_u - \Delta/2) + \beta(1 + \beta) \cdot \ell(u)
	\geq
	h(T) + \Delta/2
	\enspace,
\end{align*}
where the last inequality follows from the condition of Line~\ref{line:adding}. Similarly, if Algorithm~\ref{alg:LocalSearch} removed an element $u \in T$ from $T$ during the iteration, then, by the condition on Line~\ref{line:removing}, the new value of $h(T)$ is
\begin{align*}
	h(T - u)
	={} &
	h(T) - \beta \cdot \bE[g(u \mid T(\beta) - u)] - \beta(1 + \beta) \cdot \ell(u)\\
	\geq{} &
	h(T) - (\omega_u + \Delta/2) - \beta(1 + \beta) \cdot \ell(u)
	\geq
	h(T) + \Delta/2
	\enspace.
\end{align*}

We now argue that Algorithm~\ref{alg:LocalSearch} must reach Line~\ref{line:exit} at some point. Assume towards a contradiction that this does not happen, which by the above observation implies that the algorithm increases the value of $h(T)$ by at least $\Delta/2$ in each one of its $\lceil4n^2 / \eps + 1\rceil$ iterations. Additionally, if we denote by $T'$ the final value of the set $T$, then the initialization of $T$ and the non-negativity of $g$ guarantee together that the original value of $h(T)$ before the first iteration of Algorithm~\ref{alg:LocalSearch} is at least $\beta(1 + \beta) \cdot \ell(\{u \in \cN \mid \ell(u) > 0\}) \geq \beta(1 + \beta) \cdot \ell(T')$. Using these two results, we can lower bound the value of $h(T')$ by
\begin{align*}
	h(T')
	\geq{} &
	\beta(1 + \beta) \cdot \ell(T')
	+
	\left\lceil\frac{4n^2}{\eps} + 1\right\rceil \cdot \frac{\Delta}{2}
	>
	\beta(1 + \beta) \cdot \ell(T') + \frac{4n^2}{\eps} \cdot \frac{\Delta}{2}\\
	={} &
	\beta(1 + \beta) \cdot \ell(T') + n \cdot \max\{g(\varnothing), \max_{u \in \cN} g(u)\}
	\geq
	\beta(1 + \beta) \cdot \ell(T') + \bE[g(T'(\beta))]
	=
	h(T')
	\enspace,
\end{align*}
which is a contradiction. The strict inequality holds since $\Delta > 0$ by our assumption that $\max\{g(\varnothing),\allowbreak \max_{u \in \cN} g(u)\} > 0$, and the second inequality holds since the submodularity and non-negativity of $g$ guarantee that for every set $S \subseteq \cN$
\begin{align*}
	g(S)
	\leq{} &
	g(\varnothing) + \sum_{u \in S} g(u \mid \varnothing)
	=
	(1 - |S|) \cdot g(\varnothing) + \sum_{u \in S} g(u)
	\leq
	\max\left\{g(\varnothing), \sum_{u \in S} g(u)\right\}\\
	\leq{} &
	\max\{g(\varnothing), |S| \cdot \max_{u \in \cN} g(u)\}
	\leq
	n \cdot \max\{g(\varnothing), \max_{u \in \cN} g(u)\}
	\enspace.
\end{align*}

The last contradiction implies that our assumption was wrong, and Algorithm~\ref{alg:LocalSearch} terminates after reaching Line~\ref{line:exit}. When this happens, since the condition of Line~\ref{line:adding} evaluated to FALSE, for every element $u \in \cN \setminus T$,
\begin{align*}
	h(T)
	\geq{} &
	h(T) + \omega_u + \beta(1 + \beta) \cdot \ell(u) - \Delta\\
	\geq{} &
	h(T) + \beta \cdot \bE[g(u \mid T(\beta) - u)] + \beta(1 + \beta) \cdot \ell(u) - 3\Delta/2
	=
	h(T + u) - 3\Delta/2
	\enspace,
\end{align*}
where the second inequality holds since we assume that all the estimates made by Algorithm~\ref{alg:LocalSearch} are good. Similarly, since the condition of Line~\ref{line:removing} evaluated to FALSE, for every element $u \in T$,
\begin{align*}
	h(T)
	\geq{} &
	h(T) - \omega_u - \beta(1 + \beta) \cdot \ell(u) - \Delta\\
	\geq{} &
	h(T) - \beta \cdot \bE[g(u \mid T(\beta) - u)] - \beta(1 + \beta) \cdot \ell(u) - 3\Delta/2
	=
	h(T - u) - 3\Delta/2
	\enspace.
	\qedhere
\end{align*}
\end{proof}

The last lemma shows that with high probability the final set $T$ is an approximate local maximum with respect to $h$. Lemma~\ref{lem:good_local_max} shows that this implies that $T(\beta)$ is a good solution in expectation. To prove Lemma~\ref{lem:good_local_max}, we need the following known lemma.

\begin{lemma}[Lemma~2.2 of~\cite{feige2011maximizing}] \label{lem:2.2}
    Let $f\colon 2^X \rightarrow \nnR$ be a submodular function, and given a set $A \subseteq X$, let us denote by $A_p$ a random subset of $A$ where each element appears with probability $p \in [0, 1]$ (not necessarily independently). Then,
    \[
        \bE[f(A_p)] \geq (1-p) \cdot f(\varnothing) + p \cdot f(A)
				\enspace.
    \]
\end{lemma}

Recall that $\hat{S}$ is a subset of $\cN$ that maximizes the expression $(\alpha(\beta) - \eps) \cdot g(\hat{S}) + (\beta - \eps) \cdot \ell(\hat{S})$.
\begin{lemma} \label{lem:good_local_max}
If the set $T$ obeys
\[
	h(T) \geq h(T + u) - 3\Delta/2 \quad \forall\; u \in \cN \setminus T
	\qquad
	\text{and}
	\qquad
	h(T) \geq h(T - u) - 3\Delta/2 \quad \forall\; u \in T
	\enspace,
\]
then
\[
	\bE[g(T(\beta)) + \ell(T(\beta))]
	\geq
	(\alpha(\beta) - 3\eps/4) \cdot g(\hat{S}) + (\beta - 3\eps/4) \cdot \ell(\hat{S})
	\enspace.
\]
\end{lemma}
\begin{proof}
By the first part of Lemma~\ref{lem:Properties}, for every element $u \in \cN \setminus T$,
 \[
   h(T) \geq h(T + u) - 3\Delta/2
	\enspace,
 \] 
or equivalently $h(u \mid T) \leq 3\Delta/2$. Therefore, by the submodularity of $g$,
\begin{align} \label{eq:union}
	\bE[g(T(\beta) \cup (\hat{S} \setminus T))] &{}+ (1 + \beta) \cdot \ell(\hat{S} \setminus T)\\\nonumber
	\leq{} &
	\bE[g(T(\beta))] + \sum_{u \in \hat{S} \setminus T} \mspace{-9mu} \{\bE[g(u \mid T(\beta) - u)] + (1 + \beta) \cdot \ell(u)\}\\\nonumber
	={} &
	\bE[g(T(\beta))] + \beta^{-1} \cdot \sum_{u \in \hat{S} \setminus T} \mspace{-9mu} h(u \mid T)
	\leq
	\bE[g(T(\beta))] + 3\beta^{-1}|\hat{S} \setminus T| \Delta/2
	\enspace.
\end{align}
Similarly, since the second part of Lemma~\ref{lem:Properties} implies that for every $u \in T$ we have $h(u \mid T - u) \geq -3\Delta/2$, the submodularity of $g$ gives us
\begin{align} \label{eq:intersection}
	\bE[g(T(\beta) \cap \hat{S})] &{}- \beta(1 + \beta) \cdot \ell(T \setminus \hat{S})\\\nonumber
	\leq{} &
	\bE[g(T(\beta))] - \sum_{u \in T \setminus \hat{S}} \mspace{-9mu} \{\beta \cdot \bE[g(u \mid T(\beta) - u)] - \beta(1 + \beta) \cdot \ell(u)\}\\\nonumber
	={} &
	\bE[g(T(\beta))] - \sum_{u \in \hat{S} \setminus T} \mspace{-9mu} h(u \mid T - u)
	\leq
	\bE[g(T(\beta))] + 3|T \setminus \hat{S}| \Delta/2
	\enspace.
\end{align}
Adding $\beta$ times Inequality~\eqref{eq:union} to Inequality~\eqref{eq:intersection} now yields
\begin{align} \label{eq:combined}
	&
	\beta \cdot \bE[g(T(\beta) \cup (\hat{S} \setminus T))] + \bE[g(T(\beta) \cap \hat{S})] + \beta(1 + \beta) \cdot [\ell(\hat{S} \setminus T) - \ell(T \setminus \hat{S})]\\\nonumber
	\leq{} &
	(1 + \beta) \cdot \bE[g(T(\beta))] + 3[|\hat{S} \setminus T| + |T \setminus \hat{S}|] \Delta/2
	\leq
	(1 + \beta) \cdot \bE[g(T(\beta))] + 3n\Delta/2
	\enspace.
\end{align}

We can now use Lemma~\ref{lem:2.2} to lower bound the first two terms on the leftmost side of the last inequality as follows.
\begin{align*}
	\beta \cdot \bE[g(T(\beta) \cup (\hat{S} \setminus T))] &{}+ \bE[g(T(\beta) \cap \hat{S})]\\
	\geq{} &
	\beta(1 - \beta) \cdot g(\hat{S} \setminus T) + \beta^2 \cdot g(\hat{S} \cup T) + \beta \cdot g(T \cap \hat{S}) + (1 - \beta) \cdot g(\varnothing)\\
	\geq{} &
	\beta(1 - \beta) \cdot [g(\hat{S} \setminus T) + g(T \cap \hat{S})]
	\geq
	\beta(1 - \beta) \cdot g(\hat{S})
	\enspace,
\end{align*}
where the second inequality follows from the non-negativity of $g$, and the last inequality holds by $g$'s submodularity (and non-negativity). Plugging this inequality into Inequality~\eqref{eq:combined} now gives
\[
	\beta(1 - \beta) \cdot g(\hat{S}) + \beta(1 + \beta) \cdot [\ell(\hat{S} \setminus T) - \ell(T \setminus \hat{S})]
	\leq
	(1 + \beta) \cdot \bE[g(T(\beta))] + 3n\Delta/2
	\enspace,
\]
and rearranging this inequality yields
\begin{align*}
	\bE[g(T(\beta)&) + \ell(T(\beta))]
	=
	\bE[g(T(\beta))] + \beta \cdot \ell(T)\\
	\geq{} &
	\frac{\beta(1 - \beta) \cdot g(\hat{S}) - 3n\Delta/2}{1 + \beta} + \beta \cdot [\ell(\hat{S} \setminus T) - \ell(T \setminus \hat{S})] + \beta \cdot \ell(T)\\
	\geq{} &
	\alpha(\beta) \cdot g(\hat{S}) + \beta \cdot \ell(\hat{S}) - 3n\Delta/2
	\enspace.
\end{align*}

To complete the proof of the lemma, it remains to show that $3n\Delta/2 \leq (3\eps / 4) \cdot [g(\hat{S}) + \ell(\hat{S})]$. Towards this goal, observe that
\[
	\max_{u \in \cN} g(u)
	\leq
	\max_{u \in \cN} \left\{g(u) + \frac{1 + \beta}{2\beta^2} \cdot [\alpha(\beta) \cdot g(u) + \beta \cdot \ell(u)]\right\}
	=
	\frac{1 + \beta}{2\beta} \cdot \max_{u \in \cN} [g(u) + \ell(u)]
	\enspace,
\]
where the inequality follows from the first part of Reduction~\ref{red:remove_elements}. Using this inequality and the non-negativity of $g$, we can get
\begin{align*}
	\frac{3n\Delta}{2}
	={} &
	\frac{3\eps}{4} \cdot \max\{g(\varnothing), \max_{u \in \cN} g(u)\}
	\leq
	\frac{3\eps(1 + \beta)}{8\beta} \cdot \max\{g(\varnothing), \max_{u \in \cN} [g(u) + \ell(u)]\}\\
	\leq{} &
	\frac{3\eps(1 + \beta)}{8} \cdot [g(\hat{S}) + \ell(\hat{S})]
	\leq
	\frac{3\eps}{4} \cdot [g(\hat{S}) + \ell(\hat{S})]
	\enspace,
\end{align*}
where the penultimate inequality follows from the second part of Reduction~\ref{red:remove_elements}, and the last inequality uses the observation that the second part of Reduction~\ref{red:remove_elements} and the non-negativity of $g$ imply together that $g(\hat{S}) + \ell(\hat{S})$ is non-negative.
\end{proof}

We are now ready to prove Theorem~\ref{thm:general}.
\begin{proof}[Proof of Theorem~\ref{thm:general}]
Recall that $\hat{T}$ is a sample of $T(\beta)$ for the value of the set $T$ when Algorithm~\ref{alg:LocalSearch} terminates. Lemmata~\ref{lem:Properties} and~\ref{lem:good_local_max} prove together that there exists a high probability event $\cE$ such that
\[
	\bE[g(\hat{T}) + \ell(\hat{T}) \mid \cE]
	\geq
	(\alpha(\beta) - 3\eps/4) \cdot g(\hat{S}) + (\beta - 3\eps/4) \cdot \ell(\hat{S})
	\enspace.
\]
The last two lines of Algorithm~\ref{alg:LocalSearch} guarantee that this algorithm always outputs a set whose value is at least $g(\hat{T}) + h(\hat{T})$ because $g(\varnothing) + \ell(\varnothing) = g(\varnothing) \geq 0$. Therefore, if we denote by $\bar{T}$ the set outputted by Algorithm~\ref{alg:LocalSearch}, then we also have
\[
	\bE[g(\bar{T}) + \ell(\bar{T}) \mid \cE]
	\geq
	(\alpha(\beta) - 3\eps/4) \cdot g(\hat{S}) + (\beta - 3\eps/4) \cdot \ell(\hat{S})
	\enspace.
\]

The last two lines of Algorithm~\ref{alg:LocalSearch} also guarantee that the output set $\bar{T}$ of Algorithm~\ref{alg:LocalSearch} always has a non-negative value, and therefore, $\bE[g(\bar{T}) + h(\bar{T}) \mid \bar{\cE}] \geq 0$. Combining this inequality with the previous one using the law of total expectation yields
\begin{align*}
	\bE[g(\bar{T}) + h(\bar{T})]
	\geq{} &
	\Pr[\cE] \cdot \bE[g(\bar{T}) + h(\bar{T}) \mid \cE]\\
	\geq{} &
	(1 - o(1)) \cdot [(\alpha(\beta) - 3\eps/4) \cdot g(\hat{S}) + (\beta - 3\eps/4) \cdot \ell(\hat{S})]\\
	\geq{} &
	(\alpha(\beta) - \eps) \cdot g(\hat{S}) + (\beta - \eps) \cdot \ell(\hat{S})
	=
	\max_{S \subseteq \cN} [(\alpha(\beta) - \eps) \cdot g(S) + (\beta - \eps) \cdot \ell(S)]
	\enspace,
\end{align*}
where the second inequality holds since $g(\bar{T}) + h(\bar{T})$ is always non-negative, the equality follows from the definition of $\hat{S}$, and $o(1)$ represents a term that diminishes when $n$ goes to infinity. To justify the third inequality, note that
\[
	o(1) \cdot [(\alpha(\beta) - 3\eps/4) \cdot g(\hat{S}) + (\beta - 3\eps/4) \cdot \ell(\hat{S})]
	\leq
	o(1) \cdot (\beta - 3\eps/4) \cdot [g(\hat{S}) + \ell(\hat{S})]
	\leq
	(\eps / 4) \cdot [g(\hat{S}) + \ell(\hat{S})]
	\enspace,
\]
where the last inequality here holds for large enough values of $n$ because the second part of Reduction~\ref{red:remove_elements} and the non-negativity of $g$ imply together that $g(\hat{S}) + \ell(\hat{S})$ is non-negative.
\end{proof}

%% file: NegativeL.tex
\section{Inapproximability for Negative Linear Functions} \label{sec:negative}

In this section we prove Theorem~\ref{thm:negative_inapproximability}, which we repeat here for convenience.

\thmNegativeInapproximability*

The proof of Theorem~\ref{thm:negative_inapproximability} is based on Theorem~\ref{thm:symmetry_gap}, and therefore, we start this proof by describing an instance $\cI$ of {\RUSM}. This instance is very similar to the instance used by Oveis Gharan and Vondr\'{a}k~\cite{ovis_gharan2011submodular} to prove their hardness result for maximizing a non-negative (not necessarily monotone) submodular function subject to a matroid constraint. Specifically, the instance $\cI$ has $3$ parameters: an integer $n \geq 1$, a real value $t \geq 1$ and a real value $r \in (0, 1/2]$. The ground set of $\cI$ is $\cN = \{a, b\} \cup \{a_i, b_i \mid i \in [n]\}$, and its objective functions are $\ell(S) = -r \cdot |S \cap \{a_i, b_i \mid i \in [n]\}|$ and
\begin{align*}
	g(S)
	=
	t \cdot (|S \cap \{a, b\}| \bmod 2) &{}+ \characteristic[a \not \in S] \cdot \characteristic[S \cap \{a_i \mid i \in [n]\} \neq \varnothing] \\&+ \characteristic[b \not \in S] \cdot \characteristic[S \cap \{b_i \mid i \in [n]\} \neq \varnothing] \enspace.
\end{align*}
One can verify that $g$ is indeed a non-negative submodular function. Additionally, the functions $g$ and $\ell$ are both symmetric in the sense that the following types of swaps do not affect the values of these functions.
\begin{itemize}
	\item Any swap of the identities of the elements of $\{a_i \mid i \in [n]\}$.
	\item Swapping the identifies of $a$ with $b$ plus swapping the idenities of $a_i$ and $b_i$ for every $i \in [n]$.
\end{itemize}
Let $\cG$ be the group of permutations obtaining by combining swaps of these two kinds in any way.

In the next lemma, $G$ and $L$ are the multilinear extensions of $g$ and $\ell$, respectively, and $\bar{\vx} = \bE_{\sigma \in \cG}[\sigma(\vx)]$.
\begin{lemma} \label{lem:symmetry_gap_negative}
Let $r$ and $t$ be the values for which the maximum is obtained in the definition of $\alpha(\beta)$. Then, for any constant $\eps > 0$ and a large enough $n$, $\max_{S \subseteq \cN} [(\alpha(\beta) + \eps) \cdot g(S) + \beta \cdot \ell(S)]$ is strictly positive and
\[
	\max_{\vx \in [0, 1]^\cN} [G(\bar{\vx}) + L(\bar{\vx})] \leq \max_{S \subseteq \cN} [\alpha(\beta) \cdot g(S) + \beta \cdot \ell(S)]
	\enspace.
\]
\end{lemma}
\begin{proof}
Observe that the definition of $\cG$ guarantees that the vector $\bar{\vx}$ obeys $\bar{\vx}_a = \bar{\vx}_b$ and $\bar{\vx}_{a_i} = \bar{\vx}_{b_j}$ for every $i, j \in [n]$. Therefore, if we define for two values $z, w \in [0, 1]$ the vector $\vy(z, w)$ as follows
\[
	y_u(z, w)
	=
	\begin{cases}
		z & \text{if $u \in \{a, b\}$} \enspace,\\
		w & \text{if $u \in \{a_i, b_i \mid i \in [n]\}$} \enspace,
	\end{cases}
\]
then
\begin{align*}
	\max_{\vx \in [0, 1]^\cN} [G(\bar{\vx}) + L(\bar{\vx})]
	={} &
	\max_{z, w \in [0, 1]} [G(\vy(z, w)) + L(\vy(z, w))]\\
	={} &
	\max_{z, w \in [0, 1]} \left\{2(1 - z)[t z + 1 - (1 - w)^n] - 2rwn \right\}
	\enspace.
\end{align*}
Observe now that
\[
	(1 - w)^n
	\geq
	e^{-wn}(1 - w^2n)
	=
	e^{-wn} - \frac{(wn)^2 \cdot e^{-wn}}{n}
	=
	e^{-wn} - O(n^{-1})
	\enspace,
\]
where the last equality holds because the maximum value of the function $x^2 e^{-x}$ for $x \geq 0$ is the constant $4e^{-2}$. Plugging this observation into the previous equation now yields
\[
	\max_{\vx \in [0, 1]^\cN} [G(\bar{\vx}) + L(\bar{\vx})]
	\leq
	\max_{z, w \in [0, 1]} \left\{2(1 - z)[t z + 1 - e^{-wn}] - 2rwn \right\} + O(n^{-1})
	\enspace.
\]

The derivative of $(1 - z)[t z + 1 - e^{-wn}]$ with respect to $z$ is $t + e^{-wn} - 1 - 2t z$, which is a decreasing function of $z$ that takes the value $0$ only when $z = (t + e^{-wn} - 1)/(2t)$---note that this is a value in $[0, 1/2] \subseteq [0, 1]$. Therefore, we get from the previous inequality,
\begin{equation} \label{eq:w_only}
	\max_{\vx \in [0, 1]^\cN} [G(\bar{\vx}) + L(\bar{\vx})]
	\leq
	\max_{w \in [0, 1]} \left\{(t + 1 - e^{-wn})^2/(2t) - 2rwn \right\} + O(n^{-1})
	\enspace.
\end{equation}
The derivative of the argument of the $\max$ operation on the right hand side with respect to $wn$ (i.e., when $wn$ is treated as a single variable) is $e^{-wn}(t + 1 - e^{-wn})/t  - 2r$, which is a quadratic expression in $e^{-wn}$ whose roots are
\[
	e^{-wn}
	=
	\frac{-(t + 1) \pm \sqrt{(t + 1)^2 - 8tr}}{-2}
	=
	\frac{(t + 1) \mp \sqrt{(t + 1)^2 - 8tr}}{2}
	\enspace.
\]
One can verify that, since $r \in (0, 1/2]$, the above roots are real values, and moreover, the smaller among them falls within the range $(0, 1]$, while the larger root is at least $1$. This implies that the operand of the $\max$ operation in the right hand side of Inequality~\eqref{eq:w_only} is maximized for the $w$ value obeying
\[
	e^{-wn}
	=
	\frac{(t + 1) - \sqrt{(t + 1)^2 - 8t r}}{2}
	\enspace,
\]
and moreover, the $w$ value obeying this inequality belongs to $[0, 1]$ when $n$ is large enough. Thus,
\begin{align*}
	\max_{\vx \in [0, 1]^\cN} [G(&\bar{\vx}) + L(\bar{\vx})]
	\leq
	\frac{(t + 1 + \sqrt{(t + 1)^2 - 8t r})^2}{8t} + 2r \cdot \ln\left(\frac{t + 1- \sqrt{(t + 1)^2 - 8tr}}{2}\right) + O(n^{-1})\\
	={} &
	\frac{(t + 1)(t + 1 + \sqrt{(t + 1)^2 - 8t r})}{4t} - r \cdot \left[1 - 2\ln\left(\frac{t + 1- \sqrt{(t + 1)^2 - 8t r}}{2}\right)\right] + O(n^{-1})
	\enspace.
\end{align*}

Let us now consider the right hand side of the inequality of the lemma. Since we can choose $S = \{a, b_1\}$,
\[
	\max_{S \subseteq \cN} [(\alpha(\beta) + \eps) \cdot g(S) + \beta \cdot \ell(S)]
	\geq
	(\alpha(\beta) + \eps)(t + 1) - \beta r
	\enspace.
\]
Therefore, the inequality of the lemma holds whenever
\begin{align*}
	(\alpha(\beta) + \eps)(t + 1) - \beta r
	\geq{} &
	\frac{(t + 1)(t + 1 + \sqrt{(t + 1)^2 - 8t r})}{4t} \\&- r \cdot \left[1 - 2\ln\left(\frac{t + 1- \sqrt{(t + 1)^2 - 8t r}}{2}\right)\right] + O(n^{-1})
	\enspace,
\end{align*}
or equivalently,
\[
	\alpha(\beta) + \eps
	\geq
	\frac{t + 1 + \sqrt{(t + 1)^2 - 8t r}}{4t} - \frac{r}{t + 1} \cdot \left[1 - \beta - 2\ln\left(\frac{t + 1- \sqrt{(t + 1)^2 - 8t r}}{2}\right)\right] + O(n^{-1})
	\enspace,
\]
which is true by the definitions of $\alpha(\beta)$, $r$ and $t$ when $n$ is large enough.

To complete the proof of the lemma, it remains to argue that $\max_{S \subseteq \cN} [(\alpha(\beta) + \eps) \cdot g(S) + \beta \cdot \ell(S)] \geq (\alpha(\beta) + \eps)(t + 1) - \beta r$ is strictly positive. Since $\eps(t + 1)$ is strictly positive, it suffices to show that $\alpha(\beta) \cdot (t + 1) - \beta r$ is non-negative. By the definitions of $\alpha(\beta)$, $r$ and $t$,
\begin{equation} \label{eq:positive_proof}
	\alpha(\beta) \cdot (t + 1) - \beta r
	=
	\frac{t + 1}{4t} \cdot [t + 1 + \sqrt{(t + 1)^2 - 8t r}] - r \cdot \left[1 - 2\ln\left(\frac{t + 1- \sqrt{(t + 1)^2 - 8t r}}{2}\right)\right]
	\enspace.
\end{equation}
The derivative of the right hand side of this equality with respect to $r$ is
\begin{align*}
	&
	\frac{t + 1}{4t} \cdot \frac{-8t}{2\sqrt{(t + 1)^2 - 8t r}} - 1 + 2\ln\left(\frac{t + 1 - \sqrt{(t + 1)^2 - 8t r}}{2}\right) + 2r \cdot \frac{\frac{8t}{4\sqrt{(t + 1)^2 - 8t r}}}{\frac{t + 1- \sqrt{(t + 1)^2 - 8t r}}{2}}\\
	={} &
	-\frac{t + 1}{\sqrt{(t + 1)^2 - 8t r}} - 1 + 2\ln\left(\frac{t + 1 - \sqrt{(t + 1)^2 - 8t r}}{2}\right) \\&\qquad+ \frac{8tr}{\sqrt{(t + 1)^2 - 8t r} \cdot [t + 1- \sqrt{(t + 1)^2 - 8t r}]}
	=
	2\ln\left(\frac{t + 1 - \sqrt{(t + 1)^2 - 8t r}}{2}\right)
	\enspace.
\end{align*}
Since this derivative is an increasing function of $r$, the right hand side of Equation~\eqref{eq:positive_proof} is minimized when the derivative is $0$, i.e., when $t + 1 - \sqrt{(t + 1)^2 - 8tr} = 2$, or equivalently $r = [(t + 1)^2 - (t - 1)^2] / 8t = 1/2$. Thus, the right hand side of Equation~\eqref{eq:positive_proof} is always at least
\[
	\frac{t + 1}{4t} \cdot [t + 1 + (t - 1)] - \frac{1}{2}
	=
	\frac{t + 1}{2} - \frac{1}{2}
	=
	\frac{t}{2}
	\geq
	0
	\enspace.
	\qedhere
\]
\end{proof}

Theorem~\ref{thm:negative_inapproximability} now follows by combining Theorem~\ref{thm:symmetry_gap}, Observation~\ref{obs:negative_simplification} and Lemma~\ref{lem:symmetry_gap_negative}.

%% file: PositiveL.tex
\section{Results for Positive Linear Functions} \label{sec:positive}

In this section we study {\RUSM} in the special case in which the linear function $\ell$ is non-negative. As explained in Section~\ref{sec:introduction}, following related known results, it is natural to expect a $(1/2, 1)$-approximation for this case since $1/2$ is the best possible approximation ratio for unconstrained maximization of a non-negative submodular function. However, we show in Section~\ref{ssc:positive_inapproximability} that this cannot be done (Theorem~\ref{thm:positive_inapproximability}).

Let us now define $f \triangleq g + \ell$. As explained in Section~\ref{sec:introduction}, since $f$ is a non-negative submodular function on its own right, one can optimize it using any algorithm for \texttt{Unconstrained Submodular Maximization} (\USM). The first algorithm to obtain a tight approximation ratio of $1/2$ for {\USM} was an algorithm called ``Double Greedy'' due to Buchbinder et al.~\cite{buchbinder2015tight}. Buchbinder et al.~\cite{buchbinder2015tight} described two variants of their algorithm, a deterministic variant that we term {\DGDet} and guarantees $1/3$-approximation, and a randomized variant that we term {\DGRand} and guarantees $1/2$-approximation. It should also be noted that the original analysis of~\cite{buchbinder2015tight} proves slightly stronger results than the above stated approximation ratios. Specifically, their analysis shows that {\DGDet} always outputs a set of value at least
\[
	\frac{f(S) + f(\varnothing) + f(\cN)}{3}
	\geq
	\tfrac{1}{3}g(S) + \tfrac{2}{3}\ell(S)
\]
for any set $S$, where the inequality holds since the function $\ell$ is non-negative; which implies that {\DGDet} is a $(1/3, 2/3)$-approximation algorithm. Similarly, the analysis of Buchbinder et al.~\cite{buchbinder2015tight} shows that {\DGRand} outputs a set whose expected value is at least
\[
	\frac{2f(S) + f(\varnothing) + f(\cN)}{4}
	\geq
	\tfrac{1}{2}g(S) + \tfrac{3}{4}\ell(S)
	\enspace,
\]
which implies that {\DGRand} is a $(1/2, 3/4)$-approximation algorithm.

Theorems~\ref{thm:deterministic} and~\ref{thm:randomized} show that {\DGDet} and {\DGRand}, respectively, guarantee $(\alpha, \beta)$-approximation for many additional pairs of $\alpha$ and $\beta$. The proofs of these theorems can be found in Sections~\ref{ssc:deterministic_DG} and~\ref{ssc:randomized_DG}, respectively.

\subsection{Impossibility of the Naturally Expected Approximation Guarantee} \label{ssc:positive_inapproximability}

In this section we prove the following theorem. We note that the technique used in the proof of this theorem can also prove a somewhat stronger result. However, since the improvement represented by this stronger result is not very significant, we chose to state in the theorem the cleaner and more conceptually important result rather than the strongest result achievable.

\thmPositiveInapproximability*

Before getting to the proof of Theorem~\ref{thm:positive_inapproximability}, we need to prove the following two technical lemmata.
\begin{lemma} \label{lem:derivative_decreasing}
For every constant $c \geq 1/2$, the function $x^c \cdot \left(\frac{4}{(5 - x)^2} - 1\right)$ is a non-increasing function of $x$ for $x \in [0, 1]$.
\end{lemma}
\begin{proof}
The derivative of the function from the lemma with respect to $x$ is
\begin{align*}
	c x^{c - 1} \cdot \left(\frac{4}{(5 - x)^2} - 1\right) + x^c \cdot \frac{8}{(5 - x)^3}
	={} &
	\frac{x^{c - 1}}{(5 - x)^3} \cdot [4c(5 - x) - c(5 - x)^3 + 8x]\\
	={} &
	\frac{x^{c - 1}}{(5 - x)^3} \cdot [cx^3 - 15cx^2 + (71c + 8)x - 105c]\\
	\leq{} &
	\frac{cx^{c - 1}}{(5 - x)^3} \cdot [- 14x^2 + 87x - 105]
	\enspace.
\end{align*}
The rightmost hand side of the last inequality is always non-positive because the roots of the quadratic function $- 14x^2 + 87x - 105$ are
\[
	x_{1, 2}
	=
	\frac{-87 \pm \sqrt{87^2 - 4 \cdot 14 \cdot 105}}{2 \cdot (-14)}
	=
	\frac{87 \mp \sqrt{1689}}{28}
	\geq
	\frac{87 - \sqrt{1689}}{28}
	>
	1
	\enspace.
	\qedhere
\]
\end{proof}

\begin{lemma} \label{lem:root_lower_bound}
For every constant $x \geq 0$, $\sqrt[n]{1 - x/n} \geq 1 - O(n^{-2})$.
\end{lemma}
\begin{proof}
Observe that, for large enough $n$,
\begin{align*}
	\sqrt[n]{1 - x/n}
	={} &
	1 - \int_{1 - x/n}^1 \frac{d\sqrt[n]{y}}{dy} dy
	=
	1 - \int_{1 - x/n}^1 \frac{y^{\tfrac{1}{n} - 1}}{n} dy
	\geq
	1 - \int_{1 - x/n}^1 \frac{1}{ny} dy
	=
	1 - \frac{\left. \ln y \right|_{1 - x/n}^1}{n}\\
	={} &
	1 + \frac{\ln(1 - x/n)}{n}
	\geq
	1 - \frac{x/n}{n(1 - x/n)}
	=
	1 - \frac{x}{n(n - x)}
	=
	1 - O(n^{-2})
	\enspace.
	\qedhere
\end{align*}
\end{proof}

The proof of Theorem~\ref{thm:positive_inapproximability} is based on Theorem~\ref{thm:symmetry_gap}, and therefore, we need to describe an instance $\cI$ of {\RUSM} that has an integer parameter $n \geq 2$. The ground set of the instance $\cI$ is $\cN = \{a, b\} \cup \{c_i \mid i \in [n]\}$, and its objective functions are given, for every $S \subseteq \cN$, by $\ell(S) = 1/3$ and
\[
	g(S)
	=
	2 \cdot [(S \cap \{a, b\}) \bmod 2] + \characteristic[\{a, b\} \cap S \neq \varnothing] \cdot \characteristic[\{c_i \mid i \in [n]\} \not \subseteq S]
	\enspace.
\]
One can verify that $g$ is indeed a non-negative submodular function. Additionally, the functions $g$ and $\ell$ are both symmetric in the sense that swapping the identities of $a$ and $b$ does not change the values of these functions for any set, and the same applies to any swap of the identities of the elements of $\{c_i \mid i \in [n]\}$. Let $\cG$ be the group of permutations obtaining by combining swaps of these two kinds in any way.

In the next lemma, $G$ and $L$ are the multilinear extensions of $g$ and $\ell$, respectively, and $\bar{\vx} = \bE_{\sigma \in \cG}[\sigma(\vx)]$.
\begin{lemma} \label{lem:symmetry_gap_positive}
For a large enough $n$,
\[
	\max_{\vx \in [0, 1]^\cN} [G(\bar{\vx}) + L(\bar{\vx})] \leq \max_{S \subseteq \cN} \left[0.4998 \cdot g(S) + \frac{n - 1.0003}{n - 1} \cdot \ell(S)\right]
	\enspace,
\]
and the right hand side of the inequality is strictly positive.
\end{lemma}
\begin{proof}
Observe that the definition of $\cG$ guarantees that the vector $\bar{\vx}$ obeys $\bar{\vx}_a = \bar{\vx}_b$ and $\bar{\vx}_{c_i} = \bar{\vx}_{c_j}$ for every $i, j \in [n]$. Therefore, if we define for two values $z, w \in [0, 1]$ the vector $\vy(z, w)$ as follows
\[
	y_u(z, w)
	=
	\begin{cases}
		z & \text{if $u \in \{a, b\}$} \enspace,\\
		w & \text{if $u \in \{c_i \mid i \in [n]\}$} \enspace,
	\end{cases}
\]
then
\begin{align*}
	\max_{\vx \in [0, 1]^\cN} [G(\bar{\vx}) + L(\bar{\vx})]
	={} &
	\max_{z, w \in [0, 1]} [G(\vy(z, w)) + L(\vy(z, w))]\\
	={} &
	\max_{z, w \in [0, 1]} \left[4z(1 - z) + (2z - z^2)(1 - w^n) + \frac{nw}{3}\right]
	\enspace.
\end{align*}

Using the derivative with respect to $z$ of the argument of the $\max$ operation in the rightmost side of the last equation, one can show that the maximum is obtained when $z = 1 - 2 / (5 - w^n)$---note that this value of $z$ is indeed a number in the range $[1/2, 3/5] \subseteq [0, 1]$. Thus,
\begin{align} \label{eq:only_w_max}
	\max_{\vx \in [0, 1]^\cN} [G(\bar{\vx}) + L(\bar{\vx})]
	\leq{} &
	\max_{w \in [0, 1]} \left[\frac{8(3 - w^n)}{(5 - w^n)^2} + \frac{(3 - w^n)(7 - w^n)(1 - w^n)}{(5 - w^n)^2} + \frac{nw}{3}\right]\\ \nonumber
	={} &
	\max_{w \in [0, 1]} \left[\frac{(3 - w^n)^2}{5 - w^n} + \frac{nw}{3}\right]
	=
	\max_{w \in [0, 1]} \left[1 - w^n + \frac{4}{5 - w^n} + \frac{nw}{3}\right]
	\enspace.
\end{align}

Consider now the argument of the $\max$ operation in the rightmost side of the last inequality. The derivative of this argument with respect to $w$ is
\[
	nw^{n - 1} \cdot \left(\frac{4}{(5 - w^n)^2} - 1\right) + \frac{n}{3}
	\enspace.
\]
Let us denote the above expression by $D(w)$. Since $w \in [0, 1]$ and $\tfrac{4}{(5 - w^n)^2} - 1 \leq \tfrac{4}{16} - 1 = -\tfrac{3}{4}$, for a large enough $n$,
\[
	n(w^n)^{0.999} \cdot \left(\frac{4}{(5 - w^n)^2} - 1\right) + \frac{n}{3}
	\leq
	D(w)
	\leq
	nw^n \cdot \left(\frac{4}{(5 - w^n)^2} - 1\right) + \frac{n}{3}
	\enspace.
\]
Lemma~\ref{lem:derivative_decreasing} shows that both bounds on $D(w)$ are non-increasing functions of $w^n$. Furthermore, one can verify that the lower bound on $D(w)$ is positive for $w^n = 0.411$ and the upper bound on $D(w)$ is negative for $w^n = 0.412$. Thus, $D(w)$ is positive for $w^n \leq 0.411$ and negative for $w^n \geq 0.412$, which implies that the argument of the $\max$ operation in the rightmost side of Inequality~\eqref{eq:only_w_max} is maximized for some value $w$ such that $w^n \in [0.411, 0.412]$. Hence,
\begin{align} \label{eq:restriction}
	\max_{\vx \in [0, 1]^\cN} [G(\bar{\vx}) + L(\bar{\vx})]
	\leq{} &
	\max_{w \in [\sqrt[n]{0.411}, \sqrt[n]{0.412}]} \left[1 - w^n + \frac{4}{5 - w^n} + \frac{nw}{3}\right]\\\nonumber
	\leq{} &
	1 - 0.411 + \frac{4}{5 - 0.412} + \frac{n\sqrt[n]{0.412}}{3}
	\leq
	1.461 + \frac{n\sqrt[n]{0.412}}{3}\\\nonumber
	\leq{} &
	1.461 + \frac{n\sqrt[n]{e^{-0.886}}}{3}
	\leq
	1.461 + \frac{n\sqrt[n]{(1 - 0.886/n)^n / (1 - 0.886^2/n)}}{3}\\\nonumber
	={} &
	1.461 + \frac{n - 0.886}{3\sqrt[n]{1 - 0.886^2/n}}
	\leq
	1.1657 + \frac{n}{3(1 - O(n^{-2}))}
	\enspace,
\end{align}
where the last inequality holds by Lemma~\ref{lem:root_lower_bound}.

Since, for every value $x \in [0, 2/3]$,
\[
	\frac{n}{3(1 - x)} - \frac{n}{3}
	=
	\frac{n[1 - (1 - x)]}{3(1 - x)}
	=
	\frac{nx}{3(1 - x)}
	\leq
	nx
	\enspace,
\]
for large enough $n$, Inequality~\eqref{eq:restriction} implies
\[
	\max_{\vx \in [0, 1]^\cN} [G(\bar{\vx}) + L(\bar{\vx})]
	\leq
	1.1657 + \frac{n}{3} + n \cdot O(n^{-2})
	=
	1.1657 + \frac{n}{3} + O(n^{-1})
	\enspace.
\]

It is now time to consider the right hand side of the inequality that we need to prove. Specifically, since we can choose $S = \{a\} \cup \{c_i \mid i \in [n - 1]\}$,
\[
	\max_{S \subseteq \cN} \left[0.4998 \cdot g(S) + \frac{n - 1.0003}{n - 1} \cdot \ell(S)\right]
	\geq
	3 \cdot 0.4998 + \frac{n - 1.0003}{3}
	=
	1.4993 + \frac{n - 1}{3}
	>
	0
	\enspace.
\]
Furthermore, by combining this inequality with the previous one, we get that the inequality of the lemma holds whenever
\[
	1.4993 + \frac{n - 1}{3}
	\geq
	1.1657 + \frac{n}{3} + O(n^{-1})
	\enspace,
\]
or equivalently
\[
	0.3336
	\geq
	1/3 + O(n^{-1})
	\enspace,
\]
which is true for large enough $n$ values.
\end{proof}

By combining Theorem~\ref{thm:symmetry_gap} and Lemma~\ref{lem:symmetry_gap_positive}, we get that, even when the linear function $\ell$ is non-negative, no polynomial time algorithm for {\RUSM}  can guarantee $(0.4998(1 + \eps) + (n - 1.0003)(1 + \eps)/(n - 1))$-approximation for any $\eps > 0$ and large enough $n$. Theorem~\ref{thm:positive_inapproximability} now follows by choosing $\eps = 0.0003/n$.

\subsection{Reanalysis of Deterministic Double Greedy} \label{ssc:deterministic_DG}

In this section we prove Theorem~\ref{thm:deterministic}, which we repeat here for convenience. The algorithm {\DGDet} referred to by this theorem is given as Algorithm~\ref{alg:deterministic} (recall that $f \triangleq g + \ell$).
\thmDeterministic*

\begin{algorithm}
\caption{\DGDet} \label{alg:deterministic}
\DontPrintSemicolon
Denote the elements of $\cN$ by $u_1, u_2, \dotsc, u_n$ in an arbitrary order.\\
Let $X_0 \gets \varnothing$ and $Y_0 \gets \varnothing$.\\
\For{$i = 1$ \KwTo $n$}
{
	Let $a_i \gets f(u_i \mid X_{i - 1})$ and $b_i \gets -f(u_i \mid Y_{i - 1} - u_i)$.\\
	\lIf{$a_i \geq b_i$\label{line:b_check_det}}{Let $X_i \gets X_{i - 1} + u_i$ and $Y_i \gets Y_{i - 1}$.}
	\lElse{Let $X_i \gets X_{i - 1}$ and $Y_i \gets Y_{i - 1} - u_i$. \label{line:a_wins}}
\Return $X_n (=Y_n)$.
}
\end{algorithm}

The heart of the proof of Theorem~\ref{thm:deterministic} is the following lemma. To state this lemma, we need to define, for every integer $0 \leq i \leq n$ and set $S \subseteq \cN$, $\midset{S}{i} = (S \cup X_i) \cap Y_i$.
\begin{lemma} \label{lem:step_bound_det}
For every integer $1 \leq i \leq n$, value $\alpha \in [0, 1/3]$ and set $S \subseteq \cN$, $\alpha \cdot [f(X_i) - f(X_{i - 1})] + (1 - 2\alpha) \cdot [f(Y_i) - f(Y_{i - 1})] \geq \alpha \cdot [f(\midset{S}{i - 1}) - f(\midset{S}{i})]$.
\end{lemma}

Before we get to the proof of Lemma~\ref{lem:step_bound_det}, let us show why it implies Theorem~\ref{thm:deterministic}.
\begin{proof}[Proof of Theorem~\ref{thm:deterministic}]
Fix some $\alpha \in [0, 1/3]$ and set $S \subseteq \cN$. Summing up Lemma~\ref{lem:step_bound_det} over all integer $1 \leq i \leq n$, we get
\[
	\alpha \cdot \sum_{i = 1}^n [f(X_i) - f(X_{i - 1})] + (1 - 2\alpha) \cdot \sum_{i = 1}^n [f(Y_i) - f(Y_{i - 1})]
	\geq
	\alpha \cdot \sum_{i = 1}^n [f(\midset{S}{i - 1}) - f(\midset{S}{i})]
	\enspace.
\]
The sums in the last inequality are telescopic sums, and collapsing them yields
\[
	\alpha \cdot [f(X_n) - f(X_0)] + (1 - 2\alpha) \cdot [f(Y_n) - f(Y_0)]
	\geq
	\alpha \cdot [f(\midset{S}{0}) - f(\midset{S}{n})]
	\enspace.
\]
One can observe that $X_n = Y_n = \midset{S}{n}$, $f(X_0) = g(\varnothing) \geq 0$, $f(Y_0) = g(\cN) + \ell(\cN) \geq \ell(S)$ and $\midset{S}{0} = S$. Plugging all these observations into the previous inequality yields
\[
	\alpha \cdot f(X_n) + (1 - 2\alpha) \cdot [f(X_n) - \ell(S)]
	\geq
	\alpha \cdot [f(S) - f(X_n)]
	\enspace.
\]
It remains to rearrange the last inequality, and plug in $f(S) = g(S) + \ell(S)$, which implies
\[
	f(X_n)
	\geq
	\alpha \cdot g(S) + (1 - \alpha) \cdot \ell(S)
	\enspace.
\]
The theorem now follows since: (i) $X_n$ is the output set of Algorithm~\ref{alg:deterministic}, and (ii) the last inequality holds for every $\alpha \in [0, 1/3]$ and set $S \subseteq \cN$.
\end{proof}

Let us now prove Lemma~\ref{lem:step_bound_det}.
\begin{proof}[Proof of Lemma~\ref{lem:step_bound_det}]
Buchbinder et al.~\cite{buchbinder2015tight} showed that Algorithm~\ref{alg:deterministic} guarantees\footnote{Technically, Buchbinder et al.~\cite{buchbinder2015tight} proved Inequality~\eqref{eq:known_det} only for the special case in which $S$ is a set maximizing $f$. However, their analysis does not use this property.}
\begin{equation} \label{eq:known_det}
	[f(X_i) - f(X_{i - 1})] + [f(Y_i) - f(Y_{i - 1})]
	\geq
	f(\midset{S}{i - 1}) - f(\midset{S}{i})
	\enspace.
\end{equation}
Furthermore, we prove below that we also have the inequality
\begin{equation} \label{eq:Y_increase}
	f(Y_i) - f(Y_{i - 1})
	\geq
	0
	\enspace.
\end{equation}
These two inequalities imply the lemma together since the inequality guaranteed by the lemma is equal to $\alpha \cdot \eqref{eq:known_det} + (1 - 3\alpha) \cdot \eqref{eq:Y_increase}$---note that the coefficients $\alpha$ and $1 - 3\alpha$ in this expression are non-negative for the range of possible values for $\alpha$.

It remains to prove Inequality~\eqref{eq:Y_increase}. If $Y_i = Y_{i - 1}$, then Inequality~\eqref{eq:Y_increase} trivially holds as an equality. Consider now the case of $Y_i \neq Y_{i - 1}$. By Lines~\ref{line:b_check_det} and~\ref{line:a_wins} of Algorithm~\ref{alg:deterministic}, this case happens only when $b_i > a_i$, and $Y_i$ is set to $Y_{i - 1} - u_i$ when this happens. Therefore, we get in this case
\[
	f(Y_i) - f(Y_{i - 1})
	=
	f(Y_{i - 1} - u_i) - f(Y_{i - 1})
	=
	b_i\\
	>
	\frac{a_i + b_i}{2}
	\geq
	0
	\enspace,
\]
where the last inequality holds since Buchbinder et al.~\cite{buchbinder2015tight} also showed that $a_i + b_i \geq 0$.
\end{proof}

\subsection{Reanalysis of Randomized Double Greedy} \label{ssc:randomized_DG}

In this section we prove Theorem~\ref{thm:randomized}, which we repeat here for convenience. The algorithm {\DGRand} referred to by this theorem is given as Algorithm~\ref{alg:randomized} (recall that $f \triangleq g + \ell$).
\thmRandomized*

\SetKwIF{With}{OtherwiseWith}{Otherwise}{with}{do}{otherwise with}{otherwise}{}
\begin{algorithm}
\caption{\DGRand} \label{alg:randomized}
\DontPrintSemicolon
Denote the elements of $\cN$ by $u_1, u_2, \dotsc, u_n$ in an arbitrary order.\\
Let $X_0 \gets \varnothing$ and $Y_0 \gets \varnothing$.\\
\For{$i = 1$ \KwTo $n$}
{
	Let $a_i \gets f(u_i \mid X_{i - 1})$ and $b_i \gets -f(u_i \mid Y_{i - 1} - u_i)$.\\
	\lIf{$b_i \leq 0$\label{line:b_check}}{Let $X_i \gets X_{i - 1} + u_i$ and $Y_i \gets Y_{i - 1}$.}
	\lElseIf{$a_i \leq 0$}{Let $X_i \gets X_{i - 1}$ and $Y_i \gets Y_{i - 1} - u_i$.}
	\Else{
		\lWith{probability $\frac{a_i}{a_i + b_i}$}{Let $X_i \gets X_{i - 1} + u_i$ and $Y_i \gets Y_{i - 1}$.}
		\lOtherwise{Let $X_i \gets X_{i - 1}$ and $Y_i \gets Y_{i - 1} - u_i$.\tcp*[f]{Occurs with prob.\ $\frac{b_i}{a_i + b_i}$.}} 
	}
\Return $X_n (=Y_n)$.
}
\end{algorithm}

The heart of the proof of Theorem~\ref{thm:randomized} is the following lemma. To state this lemma, we need to define, like in Section~\ref{ssc:deterministic_DG}, $\midset{S}{i} = (S \cup X_i) \cap Y_i$ for every integer $0 \leq i \leq n$ and set $S \subseteq \cN$.
\begin{lemma} \label{lem:step_bound}
For every integer $1 \leq i \leq n$, value $\alpha \in [0, 1/2]$ and set $S \subseteq \cN$, $(\alpha/2) \cdot \bE[f(X_i) - f(X_{i - 1})] + (1 - 3\alpha/2) \cdot \bE[f(Y_i) - f(Y_{i - 1})] \geq \alpha \cdot \bE[f(\midset{S}{i - 1}) - f(\midset{S}{i})]$.
\end{lemma}

Before we get the to the proof of Lemma~\ref{lem:step_bound}, let us show why it implies Theorem~\ref{thm:randomized}.
\begin{proof}[Proof of Theorem~\ref{thm:randomized}]
Fix some $\alpha \in [0, 1/2]$ and set $S \subseteq \cN$. Summing up Lemma~\ref{lem:step_bound} over all integer $1 \leq i \leq n$, we get
\[
	\tfrac{\alpha}{2} \sum_{i = 1}^n \bE[f(X_i) - f(X_{i - 1})] + (1 - 3\alpha / 2) \cdot \sum_{i = 1}^n \bE[f(Y_i) - f(Y_{i - 1})]
	\geq
	\alpha \cdot \sum_{i = 1}^n \bE[f(\midset{S}{i - 1}) - f(\midset{S}{i})]
	\enspace.
\]
Due to the linearity of the expectation, the sums in the last inequality are telescopic sums. Collapsing these sums yields
\[
	\tfrac{\alpha}{2} \bE[f(X_n) - f(X_0)] + (1 - 3\alpha/2) \cdot \bE[f(Y_n) - f(Y_0)]
	\geq
	\alpha \cdot \bE[f(\midset{S}{0}) - f(\midset{S}{n})]
	\enspace.
\]
Observe now that, like in the proof of Theorem~\ref{thm:deterministic}, we have $X_n = Y_n = \midset{S}{n}$, $f(X_0) = g(\varnothing) \geq 0$, $f(Y_0) = g(\cN) + \ell(\cN) \geq \ell(S)$ and $\midset{S}{0} = S$. Plugging all these observations into the previous inequality yields
\[
	\tfrac{\alpha}{2} \bE[f(X_n)] + (1 - 3\alpha/2) \cdot \bE[f(X_n) - \ell(S)]
	\geq
	\alpha \cdot \bE[f(S) - f(X_n)]
	\enspace.
\]
It remains to rearrange the last inequality, and plug in $f(S) = g(S) + \ell(S)$, which implies
\[
	\bE[f(X_n)]
	\geq
	\alpha \cdot g(S) + (1 - \alpha/2) \cdot \ell(S)
	\enspace.
\]
The theorem now follows since: (i) $X_n$ is the output set of Algorithm~\ref{alg:deterministic}, and (ii) the last inequality holds for every $\alpha \in [0, 1/2]$ and set $S \subseteq \cN$.
\end{proof}

Let us now prove Lemma~\ref{lem:step_bound}.
\begin{proof}[Proof of Lemma~\ref{lem:step_bound}]
Buchbinder et al.~\cite{buchbinder2015tight} showed that Algorithm~\ref{alg:randomized} guarantees\footnote{Again, the proof of~\cite{buchbinder2015tight} was technically stated only for the case in which $S$ is a set maximizing $f$, but it extends without modification to any set $S \subseteq \cN$.}
\begin{equation} \label{eq:known}
	\bE[f(X_i) - f(X_{i - 1})] + \bE[f(Y_i) - f(Y_{i - 1})]
	\geq
	2\bE[f(\midset{S}{i - 1}) - f(\midset{S}{i})]
	\enspace.
\end{equation}
Given this inequality, to prove the lemma it suffices to show that
\[
	(1 - 2\alpha) \cdot \bE[f(Y_i) - f(Y_{i - 1})]
	\geq
	0
\]
(because adding this inequality to $\alpha/2$ times Inequality~\eqref{eq:known} yields the inequality that we want to prove). Below we prove the stronger claim that the inequality $f(Y_i) \geq f(Y_{i - 1})$ holds deterministically. One observe that this stronger claim indeed implies $(1 - 2\alpha) \cdot \bE[f(Y_i) - f(Y_{i - 1})]$ because $1 - 2\alpha$ is non-negative in the range of allowed values for $\alpha$.

If $Y_i = Y_{i - 1}$, then the inequality $f(Y_i) \geq f(Y_{i - 1})$ trivially holds as an equality. Therefore, we assume from now on $Y_i \neq Y_{i - 1}$, which implies $Y_i = Y_{i - 1} - u_i$. Due to the condition in Line~\ref{line:b_check} of Algorithm~\ref{alg:randomized}, $Y_i$ can be set to $Y_{i - 1} - u_i$ only when $b_i > 0$, and thus,
\[
	f(Y_i)
	=
	f(Y_{i - 1} - u_i)
	=
	f(Y_{i - 1}) + b_i
	>
	f(Y_{i - 1})
	\enspace.
	\qedhere
\]
\end{proof}

%% file: SymmetryGapAppendix.tex
\section{Proof of Theorem~\headerRef{thm:symmetry_gap}} \label{app:symmetry_gap}

In this section we prove Theorem~\ref{thm:symmetry_gap}, which we repeat here for convenience.
\thmSymmetryGap*

The proof of Theorem~\ref{thm:symmetry_gap} is based on the symmetry gap framework of Vondr\'{a}k~\cite{vondrak2013symmetry}. In this proof we assume $\eps < 1/2$. Note that this assumption is without loss of generality since, if Theorem~\ref{thm:symmetry_gap} applies to some constant $\eps > 0$, then it trivially holds for every larger $\eps$ value. Let us now restate two central lemmata of~\cite{vondrak2013symmetry}.\footnote{Some of the notation was modified in this restatement (compared to the original statement in~\cite{vondrak2013symmetry}) to make it easier to use these lemmata for our purposes.}

\begin{lemma}[Lemma~3.1 of~\cite{vondrak2013symmetry}] \label{lem:original_continuous_to_discrete}
Let $n$ be a positive integer, and let $F \colon [0, 1]^\cN \to \bR$ and $X = [n]$. If we define $f\colon 2^{\cN \times X} \to \nnR$ as $f(S) = F(\vx(S))$, where the vector $\vx(S)$ is defined by $x_u(S) = \frac{1}{n}|S \cap (\{u\} \times X)|$ for every $u \in \cN$. Then,
\begin{compactenum}
	\item if $\frac{\partial F}{\partial x_u} \geq 0$ everywhere for each element $u \in \cN$, then $f$ is monotone,
	\item and if the first partial derivatives of $F$ are absolutely continuous and $\frac{\partial^2 F}{\partial x_u \partial x_v} \leq 0$ almost everywhere for all elements $u, v \in \cN$, then $f$ is submodular.
\end{compactenum}
\end{lemma}

\begin{lemma}[Lemma~3.2 of~\cite{vondrak2013symmetry}] \label{lem:original_continuous_versions}
Consider a function $g\colon 2^\cN \to \nnR$ invariant under a group of permutations $\cG$ on the ground set $\cN$. Let $G(\vx)$ be the multilinear extension of $G$, define $\bar{x} = \bE_{\sigma \in \cG}[\characteristic_{\sigma(\vx)}]$ and fix any $\eps' > 0$. Then, there is $\delta > 0$ and functions $\hat{G}, \hat{H} \colon [0, 1]^{\cN} \to \nnR$ (which are also symmetric with respect to $\cG$), satisfying the following:
\begin{compactenum}
	\item For all $\vx \in [0, 1]^\cN$, $\hat{H}(\vx) = \hat{G}(\bar{\vx})$.
	\item For all $\vx \in [0, 1]^\cN$, $|\hat{G}(\vx) - G(\vx)| \leq \eps'$.
	\item Whenever $\|\vx - \bar{\vx}\|_2 \leq \delta$, $\hat{G}(\vx) = \hat{H}(\vx)$ and the value depends only on $\bar{\vx}$.
	\item The first partial derivatives of $\hat{G}$ and $\hat{H}$ are absolutely continuous.
	\item If $f$ is monotone, then, for every element $u \in \cN$, $\frac{\partial\hat{G}}{\partial x_u} \geq 0$ and $\frac{\partial\hat{H}}{\partial x_u} \geq 0$ everywhere.
	\item If $f$ is submodular then, for every two elements $u,v \in \cN$, $\frac{\partial^2\hat{G}}{\partial x_u \partial x_v} \leq 0$ and $\frac{\partial^2\hat{H}}{\partial x_u \partial x_v} \leq 0$ almost everywhere.
\end{compactenum}
\end{lemma}

In our use of Lemma~\ref{lem:original_continuous_versions} we have to carefully choose a value for the $\eps'$ parameter of the lemma. Specifically, we choose $\eps' = (\eps / 3) \cdot \max_{S \subseteq \cN}[\alpha \cdot g(S) + \beta \cdot \ell(S)]$. Notice that the conditions of Theorem~\ref{thm:symmetry_gap} guarantee that this value is (strictly) positive.

Applying Lemma~\ref{lem:original_continuous_versions} with the above chosen parameter value to the function $g$ and the group $\cG$ whose existence is guaranteed by the statement of Theorem~\ref{thm:symmetry_gap}, we get two functions $\hat{G}$ and $\hat{H}$ with the properties stated by Lemma~\ref{lem:original_continuous_versions}. Now, for every permutation $\sigma \in \cG$ and integer $n \geq 1$, we can define functions $g_1^{\sigma, n}$ and $g_2^{\sigma, n}$ as follows. For every set $S \subseteq \cN \times [n]$, let $\vy^\sigma(S)$ be the vector defined as $y^\sigma_u(S) = \tfrac{1}{n}|\{i \in [n] \mid (\sigma(u), i) \in S\}|$ for every $u \in \cN$. Then,
\[
	g_1^{\sigma, n}(S)
	=
	\hat{G}(\vy^\sigma(S))
	\qquad \text{and} \qquad
	g_2^{\sigma, n}(S)
	=
	\hat{H}(\vy^\sigma(S))
	\enspace.
\]
Observe that, by Lemma~\ref{lem:original_continuous_to_discrete}, the functions $g_1^{\sigma, n}$ and $g_2^{\sigma, n}$ are always non-negative and submodular.%, and moreover, they are monotone whenever the original function $g$ is monotone.

We now need to invoke Lemma~3.3 of~\cite{vondrak2013symmetry}. Unfortunately, the statement of this lemma is quite involved as it is designed to handle also constrained settings. Therefore, we give here only a simplified version of this lemma that suffices for our purposes.
\begin{lemma}[Special case of Lemma~3.3 of~\cite{vondrak2013symmetry}] \label{lem:indistinguishable}
Consider any deterministic sub-exponential time algorithm $ALG$ that gets access to a function $g' \colon 2^{\cN \times [n]} \to \nnR$, and let $\sigma$ be a uniformly random permutation from $\cG$. Then, with probability $1 - e^{-\Omega(n)}$, $ALG$ outputs a set of the same value when it gets either $g_1^{\sigma, n}$ or $g_2^{\sigma, n}$ as input.
\end{lemma}

Let us now construct a family of instances of {\RUSM}. For every permutation $\sigma \in \cG$, we denote by $\cI(\sigma, n)$ an instance of of {\RUSM} over the ground set $\cN \times [n]$ whose submodular and linear objective functions are $g_1^{\sigma, n}$ and $L(\vy^\sigma(S))$, respectively. We would like to prove that, when $\sigma$ is chosen uniformly at random out of $\cG$, the random instance $\cI(\sigma, n)$ is hard in expectation for every deterministic algorithm, and therefore, by Yao's principle, it is hard also for randomized algorithms. However, before doing this, let us observe that the objective functions of $\cI(\sigma, n)$ have all the necessary properties.
\begin{observation}
The function $g_1^{\sigma, n}$ is monotone whenever $g$ is, and the function $L(\vy^\sigma(S))$ is non-negative or non-positive whenever $\ell$ is non-negative or non-positive, respectively.
\end{observation}
\begin{proof}
The first part of the observation follows from Lemmata~\ref{lem:original_continuous_to_discrete} and~\ref{lem:original_continuous_versions}, and the second part of the observation holds since $L$ is the multilinear extension of $\ell$ (which implies that it only takes values that are equal to some convex combination of values taken by $\ell$).
\end{proof}

We now prove, as promised, that $\cI(\sigma, n)$ is a hard in expectation instance when the permutation $\sigma$ is chosen uniformly at random out of $\cG$.
\begin{lemma} \label{lem:symmetric_value}
Consider any deterministic sub-exponential time algorithm $ALG$ that gets the instance $\cI(\sigma, n)$ for a uniformly random $\sigma \in \cG$. Then, for a large enough $n$ (independent of $ALG$), the output set $T$ of $ALG$ obeys
\begin{align*}
	\bE[g_1^{\sigma, n}(T) + L(\vy^\sigma(T))]
	\leq{} &
	\max_{\vx \in [0, 1]^\cN} [G(\bar{\vx}) + L(\bar{\vx})] + (\eps / 2) \cdot \max_{S \subseteq \cN}[\alpha \cdot g(S) + \beta \cdot \ell(S)]\\
	\leq{} &
	(1 + \eps/2) \cdot \max_{S \subseteq \cN} [\alpha \cdot g(S) + \beta \cdot \ell(S)]
	\enspace.
\end{align*}
\end{lemma}
\begin{proof}
The second inequality of the lemma is an immediate consequence of the inequality assumed by Theorem~\ref{thm:symmetry_gap}. Therefore, we concentrate on proving the first inequality.

Since $\ell$ is a linear function, for every set $S \subseteq \cN \times [n]$,
\begin{align*}
	L(\vy^\sigma(S))
	={} &
	\sum_{u \in \cN} \ell(u) \cdot y_u^\sigma(S)
	=
	\tfrac{1}{n} \sum_{u \in \cN} \ell(u) \cdot |\{i \in [n] \mid (\sigma(u), i) \in S\}|\\
	={} &
	\tfrac{1}{n} \sum_{u \in \cN} \ell(\sigma(u)) \cdot |\{i \in [n] \mid (\sigma(u), i) \in S\}|
	=
	\tfrac{1}{n} \sum_{u \in \cN} \ell(u) \cdot |\{i \in [n] \mid (u, i) \in S\}|
	\enspace,
\end{align*}
where the penultimate equality holds since $\ell$ is invariant under $\sigma$, and the last equality holds since $\sigma$ is a permutation. This implies that the linear objective function of $\cI(\sigma, n)$ is independent of $\sigma$, and can be efficiently evaluated given $\ell$ and $n$ alone. Therefore, when $ALG$ is applied to $\cI(\sigma, n)$, we can treat the linear objective function $L(\vy^\sigma(S))$ as part of $ALG$, which makes $ALG$ an algorithm over the submodular objective function of $\cI(\sigma, n)$. Hence, by Lemma~\ref{lem:indistinguishable}, with probability $1 - e^{-\Omega(n)}$ the output set $T$ of $ALG$ obeys $g_1^{\sigma, n}(T) = g_2^{\sigma, n}(T)$.

When the equality $g_1^{\sigma, n}(T) = g_2^{\sigma, n}(T)$ holds, we can upper bound $g_1^{\sigma, n}(T)$ as follows.
\begin{align*}
	g_1^{\sigma, n}(T)
	={} &
	g_2^{\sigma, n}(T)
	=
	\hat{H}(y^\sigma(T))
	=
	\hat{G}(\overline{y^\sigma(T)})
	\leq
	G(\overline{y^\sigma(T)}) + \tfrac{\eps}{3} \cdot \max_{S \subseteq \cN} [\alpha \cdot g(S) + \beta \cdot \ell(S)]\\
	\leq{} &
	\max_{\vx \in [0, 1]^\cN} [G(\bar{\vx}) + L(\bar{\vx})] - L(\overline{y^\sigma(T)}) + \tfrac{\eps}{3} \cdot \max_{S \subseteq \cN} [\alpha \cdot g(S) + \beta \cdot \ell(S)]
	\enspace,
\end{align*}
where the third equality and the first inequality both follow from Lemma~\ref{lem:original_continuous_versions} and the value we chose for the $\eps'$ parameter of this lemma. When the equality $g_1^{\sigma, n}(T) = g_2^{\sigma, n}(T)$ does not hold, we can still observe that, since $\max_{S \subseteq \cN} [\alpha \cdot g(S) + \beta \cdot \ell(S)]$ is strictly positive by the assumptions of Theorem~\ref{thm:symmetry_gap}, there must exist a value $d$ independent of $n$ such that
\begin{align*}
	g_1^{\sigma, n}(T)
	={} &
	\hat{G}(y^\sigma(T))
	\leq
	G(y^\sigma(T)) + \tfrac{\eps}{3} \cdot \max_{S \subseteq \cN} [\alpha \cdot g(S) + \beta \cdot \ell(S)]\\
	\leq{} &
	\max_{S \subseteq \cN} g(S) + \tfrac{\eps}{3} \cdot \max_{S \subseteq \cN} [\alpha \cdot g(S) + \beta \cdot \ell(S)]
	\leq
	d \cdot \max_{S \subseteq \cN} [\alpha \cdot g(S) + \beta \cdot \ell(S)]
	\enspace,
\end{align*}
where the second inequality holds since $G$, as the multilinear extension of $g$, cannot produces values larger than the maximum value of $g$.

At this point we would like to use the law of total expectation to combine the two upper bounds on $g_1^{\sigma, n}(S)$ proved above. This leads to
\begin{align*}
	\bE[g_1^{\sigma, n}(T)]
	={} &
	\Pr[g_1^{\sigma, n}(T) = g_2^{\sigma, n}(T)] \cdot \bE[g_1^{\sigma, n}(T) \mid g_1^{\sigma, n}(T) = g_2^{\sigma, n}(T)]
	\\&+ \Pr[g_1^{\sigma, n}(T) \neq g_2^{\sigma, n}(T)] \cdot \bE[g_1^{\sigma, n}(T) \mid g_1^{\sigma, n}(T) \neq g_2^{\sigma, n}(T)]\\
	\leq{} &
	\Pr[g_1^{\sigma, n}(T) = g_2^{\sigma, n}(T)] \cdot \left\{\max_{\vx \in [0, 1]^\cN} [G(\bar{\vx}) + L(\bar{\vx})] - \bE[L(\overline{y^\sigma(T)}) \mid g_1^{\sigma, n}(T) = g_2^{\sigma, n}(T)]\right\}
	 \\&+ \tfrac{\eps}{3} \cdot \max_{S \subseteq \cN} [\alpha \cdot g(S) + \beta \cdot \ell(S)] + (1 - e^{-\Omega(n)}) \cdot d \cdot \max_{S \subseteq \cN} [\alpha \cdot g(S) + \beta \cdot \ell(S)]
	\enspace.
\end{align*}
To simplify the last inequality, we make two observations. First, that for a large enough $n$ it is guaranteed that $d(1 - e^{-\Omega(n)}) \leq \eps/6$ because both $d$ and $\eps$ are independent of $n$, and second, that the non-negativity of $g$ implies that
\[
	\bE[L(\overline{y^\sigma(T)}) \mid g_1^{\sigma, n}(T) \neq g_2^{\sigma, n}(T)]
	\leq
	\max_{x \in [0, 1]^\cN} L(\bar{x})
	\leq
	\max_{x \in [0, 1]^\cN} [G(\bar{x}) + L(\bar{x})]
	\enspace.
\]
Using these two observations and the law of total expectation (again), the previous inequality yields
\[
	\bE[g_1^{\sigma, n}(T)]
	\leq
	\max_{\vx \in [0, 1]^\cN} [G(\bar{\vx}) + L(\bar{\vx})] - \bE[L(\overline{y^\sigma(T)})] + \tfrac{\eps}{2} \cdot \max_{S \subseteq \cN} [\alpha \cdot g(S) + \beta \cdot \ell(S)]
	\enspace.
\]

The last inequality is identical to the one that we need to prove, except that the term $\bE[L(\overline{y^\sigma(T)})]$ in the last inequality should be replaced with $\bE[L(y^\sigma(T))]$. However, these two terms are identical, and therefore, the lemma follows. To see that these two terms are indeed identical, observe that, since $\ell$ is linear and invariant under the permutations of the group $\cG$,
\[
	L(\overline{y^\sigma(S)})
	=
	L\left(\bE_{\sigma' \in \cG}[\sigma'(y^\sigma(S))]\right)
	=
	\bE_{\sigma' \in \cG}[L(\sigma'(y^\sigma(S)))]
	=
	\bE_{\sigma' \in \cG}[L(y^\sigma(S))]
	=
	L(y^\sigma(S))
	\enspace.
	\qedhere
\]
\end{proof}

We are now ready to prove Theorem~\ref{thm:symmetry_gap}.
\begin{proof}[Proof of Theorem~\ref{thm:symmetry_gap}]
Consider any (possibly randomized) sub-exponential time algorithm $ALG$. By Lemma~\ref{lem:symmetric_value} and Yao's theorem, there must exist an instance $\cI(\sigma, n)$ such that $ALG$ produces a set of expected value at most $(1 + \eps/2) \cdot \max_{S \subseteq \cN} [\alpha \cdot g(S) + \beta \cdot \ell(S)]$ given this instance.

Let us now lower bound the value of the optimal solution for $\cI(\sigma, n)$. Let $T$ be the set maximizing $\max_{S \subseteq \cN} [\alpha \cdot g(S) + \beta \cdot \ell(S)]$. Then, the set $T' = \{(\sigma^{-1}(u), i) \mid u \in T, i \in [n]\}$ is a valid solution for $\cI(\sigma, n)$ such that
\begin{align*}
	\alpha \cdot g_1^{\sigma, n}(T') + \beta \cdot L(\vy^\sigma(T'))
	={} &
	\alpha \cdot \hat{G}(\vy^\sigma(T')) + \beta \cdot L(\vy^\sigma(T'))
	=
	\alpha \cdot \hat{G}(\characteristic_T) + \beta \cdot L(\characteristic_T)\\
	\geq{} &
	\alpha \cdot G(\characteristic_T) + \beta \cdot L(\characteristic_T) - \tfrac{\eps}{3} \cdot \max_{S \subseteq \cN} [\alpha \cdot g(S) + \beta \cdot \ell(S)]\\
	={} &
	\alpha \cdot g(T) + \beta \cdot \ell(T) - \tfrac{\eps}{3} \cdot \max_{S \subseteq \cN} [\alpha \cdot g(S) + \beta \cdot \ell(S)]\\
	={} &
	\left(1 - \frac{\eps}{3}\right) \cdot \max_{S \subseteq \cN} [\alpha \cdot g(S) + \beta \cdot \ell(S)]
	\enspace,
\end{align*}
where the second equality holds by the definition of $\vy^\sigma$, the inequality follows from Lemma~\ref{lem:original_continuous_versions} and the last equality holds by the definition of $T$.

Assume now towards a contradiction that $ALG$ is a $((1 + \eps)\alpha, (1 + \eps)\beta)$-approximation algorithm. Given this assumption, the above results imply together
\begin{align*}
	\left(1 + \frac{\eps}{2}\right) \cdot \max_{S \subseteq \cN} [\alpha \cdot g(S) + \beta \cdot \ell(S)]
	\geq{} &
	\max_{S \subseteq \cN} [(1 + \eps)\alpha \cdot g_1^{\sigma,n}(S) + (1 + \eps)\beta \cdot L(\vy^{\sigma}(S))]\\
	={} &
	(1 + \eps)[\alpha \cdot g_1^{\sigma, n}(T') + \beta \cdot L(\vy^\sigma(T'))]\\
	\geq{} &
	(1 + \eps) \cdot \left(1 - \frac{\eps}{3}\right) \cdot \max_{S \subseteq \cN} [\alpha \cdot g(S) + \beta \cdot \ell(S)]
	\enspace.
\end{align*}
Since one of the conditions of Theorem~\ref{thm:symmetry_gap} is that $\max_{S \subseteq \cN} [\alpha \cdot g(S) + \beta \cdot \ell(S)]$ is strictly positive, the above inequality is equivalent to
\[
	1 + \frac{\eps}{2}
	\geq
	(1 + \eps) \cdot \left(1 - \frac{\eps}{3}\right)
	=
	1 + \frac{2\eps}{3} - \frac{\eps^2}{3}
	\enspace.
\]
However, this inequality does not hold for any $\eps \in (0, 1/2)$, and thus, our assumption that $ALG$ is a $((1 + \eps)\alpha, (1 + \eps)\beta)$-approximation algorithm leads to a contradiction.
\end{proof}

%% file: UnconstrainedRegularizedMaximization.bbl
\begin{thebibliography}{10}

\bibitem{alon2000probabilistic}
Noga Alon and Joel~H. Spencer.
\newblock {\em The Probabilistic Method, Second Edition}.
\newblock John Wiley, 2000.

\bibitem{buchbinder2015tight}
Niv Buchbinder, Moran Feldman, Joseph Naor, and Roy Schwartz.
\newblock A tight linear time (1/2)-approximation for unconstrained submodular
  maximization.
\newblock {\em {SIAM} J. Comput.}, 44(5):1384--1402, 2015.

\bibitem{calinescu2011maximizing}
Gruia C{\u{a}}linescu, Chandra Chekuri, Martin P{\'{a}}l, and Jan
  Vondr{\'{a}}k.
\newblock Maximizing a monotone submodular function subject to a matroid
  constraint.
\newblock {\em {SIAM} J. Comput.}, 40(6):1740--1766, 2011.

\bibitem{conforti1984submodular}
Michele Conforti and G{\'{e}}rard Cornu{\'{e}}jols.
\newblock Submodular set functions, matroids and the greedy algorithm: Tight
  worst-case bounds and some generalizations of the rado-edmonds theorem.
\newblock {\em Discret. Appl. Math.}, 7(3):251--274, 1984.

\bibitem{dobzinski2012query}
Shahar Dobzinski and Jan Vondr{\'{a}}k.
\newblock From query complexity to computational complexity.
\newblock In Howard~J. Karloff and Toniann Pitassi, editors, {\em Proceedings
  of the 44th Symposium on Theory of Computing Conference ({STOC})}, pages
  1107--1116. {ACM}, 2012.

\bibitem{feige2011maximizing}
Uriel Feige, Vahab~S. Mirrokni, and Jan Vondr{\'{a}}k.
\newblock Maximizing non-monotone submodular functions.
\newblock {\em {SIAM} J. Comput.}, 40(4):1133--1153, 2011.

\bibitem{feldman2021guess}
Moran Feldman.
\newblock Guess free maximization of submodular and linear sums.
\newblock {\em Algorithmica}, 83(3):853--878, 2021.

\bibitem{feldman2011unified}
Moran Feldman, Joseph Naor, and Roy Schwartz.
\newblock A unified continuous greedy algorithm for submodular maximization.
\newblock In Rafail Ostrovsky, editor, {\em {IEEE} 52nd Annual Symposium on
  Foundations of Computer Science ({FOCS})}, pages 570--579. {IEEE} Computer
  Society, 2011.

\bibitem{filmus2014monotone}
Yuval Filmus and Justin Ward.
\newblock Monotone submodular maximization over a matroid via non-oblivious
  local search.
\newblock {\em {SIAM} J. Comput.}, 43(2):514--542, 2014.

\bibitem{harshaw2019submodular}
Chris Harshaw, Moran Feldman, Justin Ward, and Amin Karbasi.
\newblock Submodular maximization beyond non-negativity: Guarantees, fast
  algorithms, and applications.
\newblock In Kamalika Chaudhuri and Ruslan Salakhutdinov, editors, {\em
  Proceedings of the 36th International Conference on Machine Learning (ICML)},
  volume~97 of {\em Proceedings of Machine Learning Research}, pages
  2634--2643. {PMLR}, 2019.

\bibitem{kazemi2021regularized}
Ehsan Kazemi, Shervin Minaee, Moran Feldman, and Amin Karbasi.
\newblock Regularized submodular maximization at scale.
\newblock In Marina Meila and Tong Zhang, editors, {\em Proceedings of the 38th
  International Conference on Machine Learning ({ICML})}, volume 139 of {\em
  Proceedings of Machine Learning Research}, pages 5356--5366. {PMLR}, 2021.

\bibitem{lu2021regularized}
Cheng Lu, Wenguo Yang, and Suixiang Gao.
\newblock Regularized non-monotone submodular maximization.
\newblock {\em CoRR}, abs/2103.10008, 2021.

\bibitem{nemhauser1978best}
G.~L. Nemhauser and L.~A. Wolsey.
\newblock Best algorithms for approximating the maximum of a submodular set
  function.
\newblock {\em Mathematics of Operations Research}, 3(3):177--188, 1978.

\bibitem{nikolakaki2021efficient}
Sofia~Maria Nikolakaki, Alina Ene, and Evimaria Terzi.
\newblock An efficient framework for balancing submodularity and cost.
\newblock In Feida Zhu, Beng~Chin Ooi, and Chunyan Miao, editors, {\em The 27th
  {ACM} {SIGKDD} Conference on Knowledge Discovery and Data Mining ({KDD})},
  pages 1256--1266. {ACM}, 2021.

\bibitem{ovis_gharan2011submodular}
Shayan {Oveis Gharan} and Jan Vondr{\'{a}}k.
\newblock Submodular maximization by simulated annealing.
\newblock In Dana Randall, editor, {\em {ACM-SIAM} Symposium on Discrete
  Algorithms (SODA)}, pages 1098--1116. {SIAM}, 2011.

\bibitem{sun2022maximizing}
Xin Sun, Dachuan Xu, Yang Zhou, and Chenchen Wu.
\newblock Maximizing modular plus non-monotone submodular functions.
\newblock {\em CoRR}, abs/2203.07711, 2022.

\bibitem{sviridenko2017optimal}
Maxim Sviridenko, Jan Vondr{\'{a}}k, and Justin Ward.
\newblock Optimal approximation for submodular and supermodular optimization
  with bounded curvature.
\newblock {\em Math. Oper. Res.}, 42(4):1197--1218, 2017.

\bibitem{vondrak2013symmetry}
Jan Vondr{\'{a}}k.
\newblock Symmetry and approximability of submodular maximization problems.
\newblock {\em {SIAM} J. Comput.}, 42(1):265--304, 2013.

\end{thebibliography}
